\documentclass[12pt]{article}
\usepackage{geometry,setspace}
\usepackage{caption}
\usepackage{enumerate}
\usepackage{graphicx}
\usepackage{rotating,multirow}
\usepackage[table]{xcolor}
 \usepackage{amsmath,amsfonts,amssymb,mathrsfs,amsthm}
\usepackage{bbm}% Doublestruck maths characters like the indicator function 
\usepackage{bm}% Proper bold characters including Greek letters
\usepackage{natbib}
\usepackage{booktabs}
\usepackage{pdflscape}
\usepackage{todonotes}
\usepackage{authblk}
\usepackage{xr}
\usepackage[group-separator={,}]{siunitx}

\usepackage{xfrac}
\newcommand{\diagonehalf}{\sfrac{1}{2}}

\usepackage{caption}
\usepackage[T1]{fontenc}
\input{glyphtounicode}
\newtheorem{proposition}{Proposition}[section]
\newtheorem{theorem}[proposition]{Theorem}

\theoremstyle{definition}

\newtheorem{assumption}{Assumption}
\newcommand{\noteA}[2][noinline]{\todo[color=red!40,#1]{#2}}

\newcommand{\supplementary}{\textsc{supplementary material}}

\newcommand{\losscdf}{\ensuremath{F}}
\newcommand{\modelcdf}{\ensuremath{\widehat{\losscdf}}}

\newcommand{\PnL}{P{\&}L}
\newcommand{\lagpfunc}{\ensuremath{h}}

\newcommand{\lagpfuncmat}{\ensuremath{H}}
\newcommand{\dimW}{\ensuremath{m}}
\newcommand{\portfolio}{\textit{Pf}}
\newcommand{\betaa}{\ensuremath{a}}
\newcommand{\betab}{\ensuremath{b}}
\newcommand{\LeSt}{Lebesgue-Stieltjes}
\newcommand{\supp}{\operatorname{supp}}
\newcommand{\bigO}{\mathcal{O}}

\newcommand{\zeromean}[1]{\ensuremath{\widetilde{#1}}}
\newcommand{\onesmatrix}{\ensuremath{J}}

\newcommand{\dnudiscrete}{\ensuremath{\gamma}}  
\newcommand{\nudiscrete}{\ensuremath{\Gamma}}  
\newcommand{\likhood}{\ensuremath{{\cal L}}}  
\newcommand{\indicator}[1]{\ensuremath{\mathbbm{1}_{\{#1\}}}}
\newcommand{\imputed}[1]{\ensuremath{\hat{#1}}}

\newcommand{\EMone}{DQ}   % {$\indicator{P>.99}$}
\newcommand{\EMtwo}{V.BIN}   % {$\indicator{|2P-1|>.98}$}
\newcommand{\MDhalf}{V.\diagonehalf}  % {$\abs{2P-1}^{1/2}$}
\newcommand{\MDfour}{V.4}      % {$\abs{2P-1}^4$}

\newcommand{\average}[1]{\ensuremath{\overline{#1}}}

\newcommand{\hpower}{\ensuremath{c}}
\newcommand{\ZLp}{ZL$_{+}$}
\newcommand{\ZLn}{ZL$_{-}$}
\newcommand{\bankfiltration}{\mathcal{F}}
\newcommand{\testfiltration}{\bankfiltration^*}

\renewcommand{\P}{\mathbb{P}}
\newcommand{\E}{\mathbb{E}}

\newcommand{\N}{\mathbb{N}}

\newcommand{\var}{\operatorname{var}}
\DeclareMathOperator{\VaR}{VaR}
\newcommand{\cov}{\operatorname{cov}}

\renewcommand{\leq}{\leqslant}
\renewcommand{\geq}{\geqslant}

\newcommand{\rd}{\mathrm{d}}

\newcommand{\gi}{\leftarrow}

\newcommand{\ourThanks}{\thanks{We thank Harrison Katz and Sathya Ramesh for
     excellent research assistance.   We have
     benefitted from discussion with Mike Giles, Marie Kratz, Hsiao Yen Lok, David Lynch, David
     McArthur, Michael Milgram, and  Johanna Ziegel. The opinions
     expressed here are our own, and do not reflect the views of the
     Board of Governors or its staff.   Address correspondence to
     Michael Gordy, Federal Reserve Board, Washington DC 20551, USA,  +1-202-452-3705,
     \texttt{michael.gordy@frb.gov}.}}
\begin{document}
\title{Spectral backtests of forecast distributions\\ with application to risk 
management\ourThanks}
 \author{Michael B.~Gordy}
 \affil{Federal Reserve Board, Washington DC}
 \author{Alexander J.\ McNeil}
 \affil{The York Management School, University of York}

%\author{\mbox{ }}   % uncomment for blinded version

\date{July 26, 2019}

\maketitle
\begin{abstract}
We study a class of backtests for forecast distributions in which the test statistic  
depends on a spectral transformation that weights exceedance events by a function of the 
modeled probability level. 
The weighting scheme is specified by a kernel measure which
makes explicit the user's priorities for model 
performance.  The class of spectral backtests includes tests of unconditional coverage and 
tests of conditional coverage. 
We show how the class embeds a wide variety of backtests in the existing 
literature, and further propose novel variants which are easily implemented, well-sized
and have good power.  
In an empirical application, we backtest forecast distributions for the overnight \PnL\ of 
ten bank trading portfolios. For some portfolios, 
test results depend materially on the choice of kernel.
\end{abstract}

%\noindent \textit{Declarations of interest}: none\\
\newpage

\begin{doublespacing}
\section{Introduction}
In many forecasting exercises, fitting some range of quantiles of the forecast 
distribution may be prioritized in model design and calibration.  In risk 
management applications, which motivate this study, accuracy near the 
median of the distribution or in the ``good tail'' of high profits is generally 
much less important than accuracy in the ``bad tail'' of large losses.  Even 
within the region of primary interest, preferences may be nonmonotonic in 
probabilities.  For example, the modeller may care a great deal about assessing 
the magnitude of once-in-a-decade market disruptions, but care much less about 
quantiles in the extreme tail that are consequent to unsurvivable cataclysmic 
events.  In this paper, we study a class of backtests for forecast distributions 
in which the test statistic  
weights exceedance events by a function of the modeled probability level.  The 
weighting scheme is specified by a kernel measure which makes explicit the priorities for model 
performance.  The backtest statistic and its asymptotic 
 distribution are analytically tractable for a very large class of kernels.

Our approach unifies a wide variety of existing approaches to
backtesting.  In the area of risk management, the time-honored test
statistic \citep[dating back to][]{bib:kupiec-95} is simply a count of
``VaR exceedances,'' i.e., indicator variables equal to one whenever
the realized trading loss is in excess of the day-ahead value-at-risk
(VaR) forecast.  In our framework, this is the case where the kernel
is Dirac measure concentrated at the target VaR level.  At the other extreme, the
tests applied in \citet{bib:diebold-gunther-tay-98} represent a special case in 
which weights are uniform across all probability levels.  The likelihood-ratio
test of \citet{bib:berkowitz-01} and the expected
  shortfall and spectral risk measure tests
  of~\citet{bib:du-escanciano-17} and~\citet{bib:costanzino-curran-15}
  represent intermediate cases of a
  kernel truncated to tail probabilities. While these works are
  related  to our own, we make a distinct threefold contribution: (i)
  we offer an overarching testing framework that embeds many existing
  tests and many new ones, including discrete spectral tests and
  multivariate spectral tests; (ii) we emphasize the idea that choice of
  backtest should be guided by a user's preferences for model
  performance as expressed in kernel choice, rather than by the blind pursuit of power; 
  and (iii) we propose a general form of conditional test
  which may be combined with any kernel and which nests the
  unconditional test as a special case.
  % rather than the pursuit of an elusive ``best'' test.
  
  % number of
  % effective new tests including discrete spectral tests,
  % multispectral tests and natural extensions to conditional tests for
  % all kernels.
% represents an intermediate case of a
% kernel truncated to tail probabilities. The class of spectral
% backtests encompasses discrete kernels, which selectively weight
% forecasts at a discrete set of probability levels, as well as
% continuous kernels, which apply positive weight throughout an interval
% of levels.  Perhaps of even greater importance in practice, the class
% allows for both tests of unconditional coverage and tests of
% conditional coverage.  

The application of a weighting function in this paper bears some 
similarity to the approach of \citet{bib:amisano-giacomini-07}
and~\citet{bib:gneiting-ranjan-11} in the 
literature on comparisons of density forecasts. In both of 
those papers, weights are
applied to a forecast scoring rule to obtain measures
of forecast performance that accentuate the tails (or other
regions) of the
distribution. However, the measure for any one forecasting method has no
absolute meaning and is designed to facilitate comparison with other methods
using the general comparative testing approach proposed 
by~\citet{bib:diebold-mariano-95}. In contrast, our
tests are absolute tests of forecast quality in the spirit
of~\citet{bib:diebold-gunther-tay-98}. While the comparative testing
approach is useful for the internal refinement of the
forecasting method by the forecaster, the absolute testing approach in this 
paper facilitates external evaluation of the forecaster's results by
another agent, such as a regulator. In this paper we adopt
  the perspective of such an agent who must make a judgement based
  on a predefined set of data supplied by the forecaster and who has
  very limited information about the model choices made by the
  forecaster.

Our investigation is motivated in part by a major expansion in the data
available to regulators for the backtesting exercise.  Prior to 2013, banks in the
US reported to regulators VaR exceedances at the 99\% level. The
new Market Risk Rule mandates that banks report for each trading day the 
probability associated with the realized profit-and-loss (\PnL) 
in the prior day's forecast distribution,
which is equivalent to providing the regulator with VaR exceedances \textit{at every level}
$\alpha\in\lbrack 0,1\rbrack$.  The expanded reporting regime allows us to assess the
tradeoff between power and specificity in backtesting.   If a regulator is
concerned narrowly with the validation of reported VaR at level
$\alpha$, then a count of VaR exceedances is a sufficient statistic
for a test for unconditional coverage. However, if the regulator is
willing to assign positive weight to probability levels in a
\textit{neighborhood} of $\alpha$, we can construct more powerful
backtests.  Furthermore, our approach is consistent with a broader
view of the risk manager's mandate to forecast probabilities over a
range of large losses. The formal
guidance of US regulators to banks on internal model validation
explicitly requires ``checking the distribution of losses against
other estimated percentiles'' \citep[p.~15]{bib:us-sr-11-7}.

Under the reforms mandated by the Fundamental Review of the Trading Book 
\citep{bib:basel-13}, 99\%-VaR is replaced by 97.5\%-Expected Shortfall 
(ES) as the
determinant of capital requirements.  While there has been a lot of debate 
around the question of whether or not ES is amenable to direct backtesting 
\citep{bib:gneiting-11,bib:acerbi-szekely-14,
bib:fissler-ziegel-gneiting-16}, our contribution addresses a
  different issue. We devise tests of the \textit{forecast
    distribution} from which risk measures are estimated and not tests of the
  \textit{risk measure} estimates.  When ES is of primary
  interest it may be argued that a satisfactory forecast of the tail of the loss
  distribution is of even greater importance, since the risk measure depends
on the whole tail.

In Section \ref{sec:theory}, we lay out the statistical setting for the risk 
manager's forecasting problem and the data to be collected for backtesting.   
The transformation that
underpins the class of spectral backtests is introduced in Section \ref{sec:transf-PIT}.
Spectral backtests of unconditional coverage
are described in Section \ref{sec:tests_unconditional}.  In Section 
\ref{sec:tests_independence}, we develop tests of conditional 
coverage based on the martingale difference property.   
As an application to real data, in Section \ref{sec:empirical} we backtest 
ten bank models for overnight \PnL\ distributions for trading portfolios. 

\section{Theory and practice of risk measurement}\label{sec:theory}

We assume that a bank models \PnL\ on a filtered probability 
space $(\Omega,\mathcal{F}, (\bankfiltration_t)_{t\in\N_0},\P)$ where
$\bankfiltration_t$ represents the information available to the risk
manager at time $t$, $\N_0 = \N \cup \{0\}$ and $\N$ denotes the non-zero
natural numbers. 
For any time $t\in\N$, $L_t$ is an $\bankfiltration_{t}$-measurable random variable
representing portfolio loss (i.e., negative \PnL) in currency units.  
We denote the conditional loss distribution given information to time $t-1$ by
\begin{displaymath}
      \losscdf_{t}(x) = \P\left(L_t \leq x \mid \bankfiltration_{t-1}\right).
\end{displaymath}
The loss distribution cannot be assumed to be time-invariant.  The distribution of
returns on the underlying risk factors (e.g., equity prices, exchange rates) is 
time-varying, most notably due to stochastic volatility.  Furthermore, 
$\losscdf_t$ depends on the composition of the portfolio.  Because the
portfolio is rebalanced in each period, $\losscdf_t$ can evolve over time even
when factor returns are iid. 

For $t\in \N$ we can define the process $(U_t)$ by $U_{t} =
\losscdf_{t}(L_{t})$ using the probability integral transform (PIT).  Under the
assumption that the conditional loss distributions at each time point
are continuous, the result of~\citet{bib:rosenblatt-52}~implies that the
process $(U_t)_{t\in\N}$ is a sequence of iid standard uniform
variables, notwithstanding the fact that
$(L_t)$ is typically non-stationary. The risk manager builds a model $\modelcdf_t$ of $\losscdf_t$ based 
on information up to time $t-1$. \textit{Reported PIT-values} are the 
corresponding rvs $(P_t)$ obtained by setting $P_t =\modelcdf_{t}(L_{t})$ for
$t\in\N$.  The regulator is assumed to have no direct knowledge of $\modelcdf_t$, but can draw
inferences based on a sample of the PIT-values. If the models $\modelcdf_t$ form a sequence of \textit{ideal}
probabilistic forecasts in the sense 
of~\cite{bib:gneiting-balabdaoui-raftery-07}, i.e., coinciding with the
conditional laws $\losscdf_t$ of $L_t$ for every $t$,
then we expect the reported PIT-values to behave like an iid 
sample of standard uniform variates; tests of this property are tests
that the  sequence of models is \emph{calibrated in
  probability}.

For any $\alpha$
in the unit interval, let $\widehat{\VaR}_{\alpha,t} := \modelcdf_t^\gi(\alpha)$ be an
estimate of the $\alpha$-VaR constructed at time $t-1$ by calculating
the generalized inverse of $\modelcdf_t$ at $\alpha$.  Since the VaR exceedance event $\{  L_{t} \geq \widehat{\VaR}_{\alpha,t} \}$ is equal to the event $\{P_t\geq \alpha\}$, the PIT-value provides
a sufficient statistic for the VaR exceedances at \textit{all} possible levels.  Thus,  we would expect
well-designed tests that use reported PIT-values to be more powerful than VaR 
exceedance tests in detecting deficiencies in the models $\modelcdf_t$.

Our tests make no assumptions about the procedures and models used by the bank in
forecasting.  In practice, there is considerable heterogeneity in methodology.  
For nearly two decades, most large banks have relied primarily on some variant of 
historical sampling (HS), which is a nonparametric method based on re-sampling of
historical risk-factor changes or returns.  
As HS fails to account for serial dependencies in returns due to time-varying volatility, 
some banks adopt \textit{filtered} historical simulation (FHS) as suggested
by~\citet{bib:hull-white-98} and~\citet{bib:barone-adesi-et-al-98}. 
In this approach, the historical risk-factor returns are 
normalized by their estimated volatilities, which are typically obtained
by taking an exponentially-weighted moving-average of past squared 
returns.   Banks that do not use HS or FHS typically adopt a parametric model
for the joint distribution of risk-factor changes.

In our empirical application, testing for delayed response to changes in volatility is of special interest. 
Assuming a roughly symmetric loss distribution centered at zero, 
the frequent switching between positive and negative values will tend to cause PIT values to be serially uncorrelated, even when volatility is misspecified in
the model.  However, extreme PIT-values (i.e., near 0 or 1) will tend to beget extreme PIT-values in
high volatility periods, and middling PIT-values (i.e., near \diagonehalf) will tend to beget middling PIT-values in low volatility periods.  This pattern can be inferred by examining autocorrelation in
the transformed values $\left|2P_t - 1\right|$.  We will exploit this transformation
in implementing tests of conditional coverage in Section \ref{sec:empirical}.

There are relatively few empirical studies of bank VaR forecasting.
\citet{bib:berkowitz-o-brien-02} show that VaR estimates by US banks
are conservative (i.e., there are fewer exceedances than expected) and
that the forecasts underperform simple time-series models applied to
daily \PnL.  Conversative forecasts have been documented as well for
Canadian banks \citep{bib:perignon-deng-wang-08} and in a larger
international sample \citep{bib:perignon-smith-10}.  
The sensitivity of such results to sample period is revealed by \citet{bib:obrien-szerszen-17}. 
In their sample of five large US banks from 2001--2014, tests of unconditional coverage reject VaR
forecasts as excessively conservative for all banks in the periods of relative stability 
(2001--2006 and 2010--2014).  In the crisis period of 2007--2009, however, \citeauthor{bib:obrien-szerszen-17} 
reject VaR forecasts as insufficiently conservative for all five
banks, and serial independence is rejected for four of the banks.  This pattern is
consistent with a failure to model stochastic volatility.

\section{Spectral transformations of PIT exceedances}
\label{sec:transf-PIT}
The tests in this paper are based on transformations of indicator
variables for PIT exceedances.  The transformations take the form
\begin{equation}\label{eq:15}
  W_t = \int_{[0, 1]} \indicator{P_t \geq u} \rd \nu(u)
\end{equation}
where the \textit{kernel measure} $\nu$ is a \LeSt\ measure defined on
$[0,1]$.  The kernel measure is designed to apply weight to the probability levels of greatest
interest, typically (in practice) in the region of the standard VaR level 
$\alpha=0.99$.   With any \LeSt\ measure $\nu$ on domain $[0,1]$, there is an associated 
increasing right-continuous function $G_\nu$ such that $\nu([0,u]) = G_\nu(u)$. 
It is easily seen that \eqref{eq:15} is equivalent to the closed-form expression 
\begin{equation}\label{eq:W-closed-form}
  W_t = \nu([0,P_t]) = G_\nu(P_t)
\end{equation}
which shows that $W_t$ is increasing in $P_t$.   The measure can be normalized
such that $G_\nu(1)=1$ without loss of generality, but we do not require it.
To streamline the presentation, we will
henceforth impose the following mild regularity condition on $\nu$.
\begin{assumption}
\label{ass:almostcontinuous}
$\nu(\{0\})=\nu(\{1\})=0$ and $G_\nu$ is differentiable except at a finite set of
 points. 
\end{assumption}

The kernel measure can be discrete, continuous or mixed.  In the discrete case, it takes the 
form $\nu = \sum_{i=1}^m \dnudiscrete_i
\delta_{\alpha_i}$ for $m\geq 1$ where $\delta$ denotes Dirac
measure.  This places positive mass $\dnudiscrete_1,\ldots, \dnudiscrete_m$ at the
ordered values $0<\alpha_1<\cdots<\alpha_m<1$ leading to
\begin{equation}\label{eq:16}
  W_{t} = \sum_{i=1}^m \dnudiscrete_i \indicator{P_t \geq \alpha_i}.
\end{equation}
For the continuous case, the measure has density $\rd\nu(u) =  g_\nu(u) \rd u$ for
some nonnegative $g_\nu(u)$ defined on $[0, 1]$ which we refer to as the \textit{kernel density}. 
The univariate transformation extends naturally to the 
multivariate case in which a set of distinct kernel
measures $\nu_1, \ldots,\nu_\dimW$ is applied to PIT-values to obtain
the vector-valued variables $\bm{W}_1\,\ldots,\bm{W}_n$ where
\begin{equation}
  \label{eq:20}
  \bm{W}_t = (W_{t,1},\ldots,W_{t,\dimW})^\prime, \quad W_{t,j} 
   = \nu_j([0,P_t]) = G_j(P_t),\; j=1,\ldots,\dimW.
\end{equation}

%\subsection{Spectral backtests}\label{sec:spectral-backtests}
We will refer to any backtest based on spectrally transformed PIT
exceedances as a \textit{spectral backtest}.  For the purposes of this paper, we assume that
the regulator can utilize only present and past values 
of $P_t$ in the backtest statistic.  This restriction could be relaxed considerably.\footnote{Our approach could easily
   be generalized to incorporate information in $(L_t,
   \widehat{\VaR}_{\alpha,t})$ and in publicly observed market
   variables (such as VIX). However, frequent change in
  portfolio composition implies that lagged VaR values are less reliably 
  informative than lagged PIT values.} What is essential to our contribution is that the
 regulator does not observe the entire distribution $\hat{F}_t$, but
 does observe more than the VaR exception indicator $\indicator{L_t
 \geq \widehat{\VaR}_{\alpha,t}}$.  
 
 Let $(\testfiltration_t)$ be the regulator's filtration generated by the 
 PIT values, i.e.,~$\testfiltration_t =\sigma(\{P_s: s\leq t\}) \subset
 \bankfiltration_t$.  Regardless of the form of the test, 
the null hypothesis is 
\begin{equation}
  \label{eq:nullhypothesis}
 H_0:\quad \bm{W}_t \sim F_{W}^0\; \text{and}\; \bm{W}_t \perp\!\!\!\perp
   \testfiltration_{t-1}, \; \forall t,
 \end{equation}
 where $F_{W}^0$ denotes the distribution function of
 $\bm{W}_t$ when $P_t$ is uniform.  The null
hypothesis~\eqref{eq:nullhypothesis} implies that the
$(\bm{W}_t)$ are iid but  
is weaker than a null hypothesis that the $(P_t)$ are iid Uniform.  This is by intent.  
Since the regulator is free to choose $\nu$ in accordance with her priorities, she should not object to 
departures from uniformity and serial independence that arise outside
the support of her chosen kernel.

%We will refer to any backtest based on spectrally transformed PIT
%exceedances as a \textit{spectral backtest}.   Regardless of the form of the test,
%the null hypothesis is \noteA{new text}
%\begin{equation}
%  \label{eq:nullhypothesis}
% H_0:\quad \bm{W}_t \sim F_{W}^0\; \text{and}\; \bm{W}_t \perp\!\!\!\perp
%   \testfiltration_{t-1}, \; \forall t,
% \end{equation}
% where $F_{W}^0$ denotes the distribution function of
% $\bm{W}_t$ when $P_t$ is uniform, and where
% $(\testfiltration_t)$ is the filtration generated by the PIT values, i.e.~$\testfiltration_t =\sigma(\{P_s: s\leq t\}) \subset
% \bankfiltration_t$. For the purposes of this paper, we assume that
% the regulator’s tests can condition only on present and past values
% of $P_t$.  What is essential to our contribution is that the
% regulator does not observe the entire distribution $\hat{F}_t$, but
% does observe more than the VaR exception indicator $\indicator{L_t
%   \geq \widehat{\VaR}_{\alpha,t}}$\footnote{Our approach could easily
%   be generalized to incorporate information in $(L_t,
%   \widehat{\VaR}_{\alpha,t})$ and in publicly observed market
%   variables (such as VIX). However, frequent change in
%  portfolio composition implies that lagged VaR values are usually less
%  informative than lagged PIT values.}.
 
%There is a large literature on correcting tests of forecasts for estimation error,
%  including~\citet{bib:escanciano-olmo-10},~\citet{bib:du-escanciano-17}
%  and~\citet{bib:hurlin-et-al-17}.
Several recent papers propose to correct tests of forecasts for estimation error; see, e.g., \citet{bib:escanciano-olmo-10,bib:du-escanciano-17,bib:hurlin-et-al-17}.  Implementation of these corrections generally requires
  knowledge of the forecasting model and estimation scheme
  and is thus infeasible in the regulatory context we describe.  Our null hypothesis imposes the high standard that the forecaster is an ideal forecaster working with a sequence of correctly specified, perfectly estimated models.

The spectral class encompasses a great variety of tests but we prioritize two
general testing approaches: Z-tests and likelihood ratio (LR) tests. In the univariate 
case, the spectral Z-test is based on the asymptotic normality of $\average{W}_n = 
n^{-1}\sum_{t=1}^n
  W_t$ under the null hypothesis~\eqref{eq:nullhypothesis}. 
  Writing $\mu_W =\E(W_t)$ and
  $\sigma^2_W = \var(W_t)$ for the moments in the null
  model $F_W^0$, it follows from the central limit theorem that
  \begin{equation}\label{eq:2}
  Z_n = \frac{\sqrt{n}(\average{W}_n - \mu_W)}{\sigma_W} \xrightarrow[n\to 
\infty]{d} N(0,1).
  \end{equation}
In the multivariate case ($\operatorname{dim} \bm{W}_t =\dimW$) we have
%\begin{displaymath}
 $ \sqrt{n}\left(\average{\bm{W}}_n-\bm{\mu}_W\right) \xrightarrow[n\to 
\infty]{d} N_\dimW(\bm{0},\Sigma_W)$
%\end{displaymath}
where $\average{\bm{W}}_n = n^{-1}\sum_{t=1}^n
  \bm{W}_t$ and $\bm{\mu}_W$ and $\Sigma_W$ are the mean vector and covariance
matrix of the null distribution $F_{W}^0$.
Hence a test can be based on assuming for large enough $n$ that
  \begin{equation}\label{eq:28}
    T_n = n\left(\average{\bm{W}}_n-\bm{\mu}_W\right)^\prime
    \Sigma_W^{-1} \left(\average{\bm{W}}_n-\bm{\mu}_W\right) \sim \chi^2_\dimW,
  \end{equation}
where we refer to $T_n$ as an $\dimW$-spectral Z-test statistic.

The first moment of the transformed PIT-values under the null hypothesis is easily obtained as
\begin{equation}
\mu_W = \int_{[0,1]} (1-u) \rd\nu(u)
\label{eq:muW}
\end{equation}
The variance $\sigma_W^2$ and the cross-moments in $\Sigma_W$ are obtained using
a simple product rule for spectrally transformed PIT values.
\begin{theorem}\label{prop:product-of-Ws}
The set of spectrally transformed PIT values
defined by $W_{t,j}  = \nu_j([0,P_t])$ is closed under multiplication.  The product
$W^*_t =  W_{t,1} W_{t,2}$ is given by $W^*_t = \nu^*([0,P_t])$ where 
  $\nu^*$ is a \LeSt\ measure which satisfies
%\begin{equation}\label{eq:nustar}
\[
\nu^*([0,u]) = 
\int_{[0,u]} \Big(\nu_2([0,s])   - \frac{1}{2}\nu_2(\{s\})\Big)  \rd\nu_1(s)
+
\int_{[0,u]}\Big( \nu_1([0,s]) - \frac{1}{2}\nu_1(\{s\})\Big) \rd  \nu_2(s). 
\]
%\end{equation}
\end{theorem}
It follows that $\sigma_W^2 =\mu_{W^*}-\mu_W^2$, where $\mu_{W^*}$ is found by applying
\eqref{eq:muW} under the measure $\nu^*$ obtained when
$\nu_1=\nu_2=\nu$. This yields
\begin{equation}
\label{eq:mustar}
  \mu_{W^*} = \int_{[0,1]}(1-u)\left(2G_\nu(u)
    -\nu(\{u\})\right)\rd \nu(u)\,.
\end{equation}

The central limit theorem underpinning the Z-test requires finite second moments.  For the univariate
case, the following proposition provides a sufficient condition on the tail behavior of $G_\nu$.
\begin{proposition} \label{prop:tailbehavior}
If $G_\nu(u)=\bigO((1-u)^{-0.5+\epsilon})$ as $u\rightarrow 1$ for some $\epsilon>0$, then
$\sigma_W^2$ is finite.
\end{proposition}
\noindent In the multivariate setting, the asymptotic distribution in \eqref{eq:28} holds if the
condition in Proposition \ref{prop:tailbehavior} are satisfied for each $\nu_j$, $j=1,\ldots,\dimW$.

Likelihood ratio tests are based on continuous parametric models $F_P(\cdot\mid\bm{\theta})$ for the PIT values
$P_t$ that nest uniformity as a special case corresponding to $\bm{\theta}=\bm{\theta}_0$. The implied model 
$F_W(\cdot\mid\bm{\theta})$ for the values $W_t=G_\nu(P_t)$ is used to test
the null hypothesis~\eqref{eq:nullhypothesis} with $F_W^0 =
F_W(\cdot\mid\bm{\theta}_0)$; the alternative is that $\bm{W}_t \sim
F_W(\cdot\mid\bm{\theta})$ with $\bm{\theta}\neq\bm{\theta}_0$. Writing $\likhood_W(\bm{\theta}\mid
     \bm{W})$ for the likelihood function, the test is based on the
     asymptotic chi-squared distribution of the statistic
     \begin{equation}\label{eq:3}
       \text{LR}_{W,n} = \frac{\likhood_W(\bm{\theta}_0\mid
     \bm{W})}{\likhood_W(\hat{\bm{\theta}}\mid\bm{W})}
     \end{equation}
where $\hat{\bm{\theta}}$ denotes the maximum likelihood estimate based on the transformed sample $(\bm{W}_t)$.

An important difference between the two classes of test is that the
Z-test is sensitive to the choice of kernel whereas
the LR-test is sensitive only to the support of the kernel. Considering the univariate case for simplicity,
we show 
\begin{theorem}\label{theorem:spectr-transf-pit}
Let $\nu_1$ and $\nu_2$ be \LeSt\ measures satisfying 
Assumption \ref{ass:almostcontinuous}, and let $W_{t,j}=G_j(P_t)$ for $j=1,2$ and $t=1,\ldots,n$ be the
respective samples of transformed PIT values. If $\supp(\nu_1) =
\supp(\nu_2)$ then $\operatorname{LR}_{W_1,n}=\operatorname{LR}_{W_2,n}$
almost surely. 
\end{theorem}
\noindent This result can be viewed as a generalization of the invariance property of LR-tests under 
one-to-one transformations of the data.

\section{Tests of unconditional coverage}
\label{sec:tests_unconditional}
It is common to divide backtesting methods into tests of unconditional
coverage and tests of conditional coverage. 
% In the context of VaR
% backtesting, an unconditional test is a test that exceedances are
% Bernoulli events with the correct probability of occurrence while a
% conditional test is a test that exceedances have the correct
% conditional probability of occurrence, which is equivalent to
% requiring that they are also independent events. 
% For
In our setting, an unconditional test is a test for the distribution $F_W^0$ implied by the 
uniformity of the PIT-values while a conditional
test is a test for both the correct distribution and the
independence of $\bm{W}_t$ and $\testfiltration_{t-1}$ for all
$t$. 

In this section we present a number of unconditional tests based on
the Z-test and LR-test ideas discussed in Section~\ref{sec:transf-PIT}.  
It is important to note that the convergence results on which these tests are 
based, although mostly stated under iid
assumptions, do hold in situations where the independence assumption
is relaxed, for example for stationary and ergodic
martingale-difference processes (according to the martingale CLT of
~\citet{bib:billingsley-61}).  In the case of the univariate Z-test, the test 
will have no power to detect serial
dependence whenever $\lim_{n\to\infty}\var(\sqrt{n}\average{W}_n) \approx \sigma^2_W$.
If, however, there is persistent positive serial
correlation in $(W_t)$ leading to $\lim_{n\to\infty}\var(\sqrt{n}\average{W}_n) 
>\sigma^2_W$ then the Z-test will have some power to detect
dependencies; however, more targeted tests of the independence property
are available and are the subject of
Section~\ref{sec:tests_independence}.

Our unconditional testing approach subsumes a number of important published
tests or close relatives thereof. The discrete weighting framework
in Section~\ref{sec:discrete-weighting}
includes the binomial LR-test of~\citet{bib:kupiec-95} and~\citet{bib:christoffersen-98}
for the number of VaR exceedances. It also includes the multilevel
Pearson chi-squared test recommended by~\citet{bib:campbell-06} and the multilevel
LR-test proposed in ~\citet{bib:perignon-smith-08} which
also underlies the work of~\citet{Colletaz2013}; see~\citet{bib:kratz-lok-mcneil-18}
for a comparative study of multilevel tests.

The continuous weighting framework in
Section~\ref{sec:continuous-weighting} builds on the seminal paper of
\citet{bib:diebold-gunther-tay-98}. It is close in spirit to the approach of~\citet{bib:crnkovicDrachman-96}, who apply a statistic based on a
weighted Kuiper distance between the distribution of PIT values and
the uniform, and subsumes the likelihood ratio test of \citet{bib:berkowitz-01} based on fitting a truncated normal distribution
to probit-transformed PIT-values.
Most closely related to our work, 
\cite{bib:du-escanciano-17} and \citet{bib:costanzino-curran-15} have proposed test statistics for
spectral risk measures which can be viewed as special cases of our
univariate spectral Z-test.
% Both papers focus on the case of a uniform
% kernel and interpret the tests in terms of backtesting expected
% shortfall. 

\subsection{Discrete weighting}\label{sec:discrete-weighting}
Discrete tests are based on the univariate transformation
$W_{t} = \sum_{i=1}^m \dnudiscrete_i \indicator{P_t \geq \alpha_i}$ 
as defined in~\eqref{eq:16} and the multivariate transformation  
$\bm{W}_t = (\indicator{P_t \geq \alpha_1},\ldots, \indicator{P_t \geq  \alpha_m})^\prime$ 
in~(\ref{eq:20}) for the same set of ordered levels $\alpha_1 < \cdots <  \alpha_m$.
Obviously, when $m=1$ (and $\dnudiscrete_1=1$) both transformations yield $W_t = \indicator{P_t \geq \alpha}$, so that we
obtain iid Bernoulli$(1-\alpha)$
variables under the null hypothesis~\eqref{eq:nullhypothesis}.  This is the 
basis for standard VaR exceedance testing based on the binomial
distribution. The Z-test statistic~(\ref{eq:2}) for $W_t =
\indicator{P_t \geq \alpha}$ coincides with the binomial score test statistic
\begin{equation}\label{eq:53}
 Z_n =  \frac{\sqrt{n} 
\left(\average{W}_n-(1-\alpha)\right)}{\sqrt{\alpha(1-\alpha)}}.
\end{equation}
The LR-test uses an implicit nesting model for $P_t$ in which the $W_t$ are
iid Bernoulli($p$) and tests $p=1-\alpha$ against $p\neq 1-\alpha$ by
comparing the statistic~\eqref{eq:3}  to a $\chi^2_1$ 
distribution;  this is the approach taken in~\citet{bib:kupiec-95}
and~\citet{bib:christoffersen-98}. 

 When $m>1$ the variables $W_t = \sum_{i=1}^m
 \dnudiscrete_i \indicator{P_t \geq \alpha_i}$ take the ordered values
$\nudiscrete_0 <\nudiscrete_1<\cdots <\nudiscrete_m$ where $\nudiscrete_0 = 0$ and 
$\nudiscrete_k = \sum_{i=1}^k \dnudiscrete_i$ for
$k=1,\ldots,m$. Under the null hypothesis~\eqref{eq:nullhypothesis} the distributions of $W_t$
and $\bm{W}_t$ satisfy
\begin{equation}\label{eq:26}
  \P(W_{t} = \nudiscrete_i) = \P(\bm{1}^\prime\bm{W}_t = i) = \alpha_{i+1}-\alpha_i,\quad i\in \{0,1,\ldots,m\},
\end{equation}
where $\alpha_0=0$ and $\alpha_{m+1} =1$. In both cases this describes
a multinomial distribution.

The univariate and multivariate tranformations result in
different Z-tests which can be considered as alternative
generalizations of the binomial score test~\eqref{eq:53}.  Application of
Theorem \ref{prop:product-of-Ws} to the univariate case and use of~(\ref{eq:mustar})
delivers moments under the null  $\mu_W = \sum_{i=1}^m \dnudiscrete_i(1-\alpha_i)$ and 
$\sigma^2_W =  \sum_{i=1}^m \dnudiscrete_i^* (1-\alpha_i)-\mu_W^2$
where $\dnudiscrete_i^* =(2\nudiscrete_{i} - \dnudiscrete_i)\dnudiscrete_i$.
In constructing the test statistic $Z_n$ in~\eqref{eq:2}, we can vary the
weights $\dnudiscrete_i$ to emphasize different levels $\alpha_i$
and obtain a variety of new tests.

In the multivariate case, we construct an $m$-spectral Z-test as in~\eqref{eq:28} with
$\bm{\mu}_W=(1-\alpha_1,\ldots,1-\alpha_m)^\prime$ and second moment matrix
$\Sigma_W$ with $(i,j)$ element given by  $\alpha_{i \wedge j} (1-\alpha_{i\vee  j})$. 
We then  obtain the 
classical Pearson chi-squared statistic as proposed by~\cite{bib:campbell-06}.
\begin{theorem}\label{theorem:pearson-test}
  \begin{displaymath}
    n (\average{\bm{W}}_n -\bm{\mu}_W)^\prime
\Sigma_W^{-1} (\average{\bm{W}}_n -\bm{\mu}_W) = \sum_{i=0}^m
\frac{(O_{i} - n\theta_i)^2}{n\theta_i} 
  \end{displaymath}
where $O_i = \sum_{t=1}^n \indicator{\bm{1}^\prime \bm{W}_t =i}$ and
$\theta_i = \alpha_{i+1}-\alpha_i$ for $i=0,\ldots,m$.
\end{theorem}

To implement a multinomial (or multi-level) LR-test of~(\ref{eq:26})
we use a nesting model for $P_t$ in
which  $\P(W_{t} = \nudiscrete_i) = \P(\bm{1}^\prime\bm{W}_t = i) =
p_i$ and $\sum_{i=0}^m p_i =1$. The likelihoods based on $(W_t)$ and
$(\bm{W}_t)$ yield the same
sufficient statistics $O_i = \sum_{t=1}^n \indicator{W_t = \Gamma_i}
=\sum_{t=1}^n \indicator{\bm{1}^\prime\bm{W}_t =i}$ for the cell
probabilities $p_i$. By the likelihood principle the univariate and multivariate LR-tests are
identical and depend only on the levels $(\alpha_1,\ldots,
\alpha_m)$ and not the weights $\dnudiscrete_i$ in
the univariate transformation. The invariance of the univariate
LR-test under different choices for the weights is also a consequence of Theorem \ref{theorem:spectr-transf-pit}.

\subsection{Continuous weighting}\label{sec:continuous-weighting}

Consider a kernel measure $\nu$ with associated absolutely
continuous $G_\nu$.  We assume that the kernel density satisfies $g_\nu(u)>0$ for
$\alpha_1<u<\alpha_2$ and $g_\nu(u) =0$
for $u<\alpha_1$ and $u > \alpha_2$. We refer to $\supp(\nu) = [\alpha_1,\alpha_2]$
as the \textit{kernel window}.

When $\nu_1$ and $\nu_2$ are both continuous kernels with the same
kernel window, Theorem~\ref{prop:product-of-Ws}
simplifies.  The kernel $\nu^*$ for the product $W^*_t=W_{t,1}W_{t,2}$ is continuous on the same kernel
window with density
$g^*(u) = G_1(u)g_2(u) + G_2(u)g_1(u)$.  Moments and cross-moments can be obtained analytically 
for a wide variety of kernel densities, e.g., based on polynomials, exponential functions, or on 
beta-type densities of the form $(u-\alpha_1)^{\betaa-1}(\alpha_2-u)^{\betab-1}$
for $\betaa,\betab>0$; see Section \ref{sec:mcUnconditional} for
examples of new tests based on this idea. 
Thus, our compact presentation of the continuous spectral Z-test subsumes
a very large class of possible tests.

For the LR-test, recall that we require a family  of distributions $F_P(\cdot\mid\bm{\theta})$ for the PIT values
that nests uniformity as a special case corresponding to
$\bm{\theta}=\bm{\theta}_0$.
Since $\supp(\nu)=[\alpha_1,\alpha_2]$ for all the kernels of this section,
Theorem~\ref{theorem:spectr-transf-pit} implies that they all give
rise to identical LR-tests, depending only on the kernel window
$[\alpha_1,\alpha_2]$ and the nesting model
$F_P(\cdot\mid\bm{\theta})$.  The form taken by $G_\nu$ on $[\alpha_1,
\alpha_2]$ is immaterial.

Drawing upon the probitnormal model, we assume that the
PIT values $P_1,\ldots,P_n$ have a distribution
satisfying $\Phi^{-1}(P_t) \sim N(\mu,\sigma^2)$.
Writing $\bm{\theta}=(\mu,\sigma)^\prime$, the
distribution function and density of $P_t$ are respectively
\begin{equation}\label{eq:17}
 F_P(p\mid \bm{\theta}) = \Phi\Bigg(\frac{\Phi^{-1}(p)-\mu}{\sigma}\Bigg),\quad
 f_{P}(p\mid \bm{\theta})
=\frac{\phi\Big(\frac{\Phi^{-1}(p)-\mu}{\sigma}\Big)}{\phi(\Phi^{-1}(p))\sigma},
\quad p \in [0,1],
\end{equation}
and the uniform distribution corresponds to
$\bm{\theta}=\bm{\theta}_0=(0,1)^\prime$.  The test of~\cite{bib:berkowitz-01} is a
special case of the LR-test under this nesting model: choosing the kernel
$\nu$ with $G_\nu(u) = \Phi^{-1}(\alpha_1 \vee u) - \Phi^{-1}(\alpha_1)$ and kernel window
$[\alpha_1,1]$, we observe that $W_t = G_\nu(P_t)$ has a
$N(\mu,\sigma^2)$ distribution truncated to
$[\Phi^{-1}(\alpha_1),\infty)$ and then translated to the left. 

The probitnormal model truncated to
the window $[\alpha_1,\alpha_2]$ also yields a further new bispectral
Z-test based on a classical score test of $\bm{\theta}=\bm{\theta}_0$ against the alternative
$\bm{\theta}\neq\bm{\theta}_0$. Full details of this test are found in
our companion paper [\supplementary]. The resulting \textit{truncated
  probitnormal score test} 
has a mixed weighting scheme (with discrete and continuous parts) in which
\begin{equation}\label{eq:9}
  W_{t,i} = \dnudiscrete_{i,1}\indicator{P_t \geq \alpha_1} + \dnudiscrete_{i,2}\indicator{P_t \geq
  \alpha_2} + \int_{\alpha_1}^{\alpha_2}g_i(u)\indicator{P_t\geq u}\rd
u,\quad i=1,2,
\end{equation}
for known constants $\dnudiscrete_{i,1}\geq 0$,
$\dnudiscrete_{i,2}\geq 0$ and known functions
$g_i(u)$ which are positive and differentiable on $[\alpha_1, \alpha_2]$.

\noteA[inline]{NOTE:
For completeness and to preserve blinding, material drawn from the
``companion paper'' is
appended to this paper as ``\supplementary''.}

\subsection{Size and power}
\label{sec:mcUnconditional}
We have performed extensive Monte Carlo analyses of the size and power of unconditional spectral 
backtests. In this section, we offer representative examples which illustrate how the size and power of Z-tests 
depend on the kernel, and then
briefly summarize other key findings.  Full details of all simulation
experiments may be found in the companion paper [\supplementary].

We consider kernels of discrete, continuous and mixed form.  
Parameters $\alpha_1$ and $\alpha_2$ control the kernel window. For the continuous
tests, $\alpha_1$ and $\alpha_2$ are the infimum and supremum of
the kernel support. For the discrete case, we consider 3-level kernels at the set of points 
$(\alpha_1, \alpha^*, \alpha_2)$, where $\alpha^*=0.99$ is the conventional VaR level. 
We define a \textit{narrow} window for which $\alpha_1=0.985$ and  $\alpha_2=0.995$,
and a \textit{wide} window for which $\alpha_1=0.95$ and  $\alpha_2=0.995$.  Observe that
the narrow window is symmetric around $\alpha^*$, whereas the wide window is asymmetric.

For the continuous case, there is a wide variety of plausible candidates for the kernel.  
Table \ref{table:kerneldensity} lists the kernels that we discuss below; each may be thought of 
as describing a family of kernel densities for different windows $[\alpha_1, \alpha_2]$.  For parsimony, 
all are special cases of the beta kernel.  The uniform and hump-shaped 
Epanechnikov kernels are commonly used in the nonparametric statistics 
literature. In the \supplementary,
we provide analytical solutions for the moments of transformed PIT
values for the general beta$(\betaa,\betab)$ case.
\begin{table}[htbp]
  \centering
  \begin{tabular}{r l l c}
    \toprule
   Kernel family & Mnemonic & Density $g(u)$ & Beta representation \\  \midrule
   Uniform & ZU & $1$ & 1,1 \\
   Arcsin & ZA & $1/\sqrt{u^*(1-u^*)}$ & \diagonehalf,\diagonehalf\\
   Epanechnikov & ZE & $1-(2u^*-1)^2$ & 2,2 \\
   Linear increasing & ZL$_{+}$ & $u^*$ & 2,1 \\
   Linear decreasing & ZL$_{-}$ & $1-u^* $ & 1,2 \\
    \bottomrule
  \end{tabular}
  \caption{Kernel density functions on $[\alpha_1,\alpha_2]$.\\
    $u^*$ denotes the rescaled value $u^*=(u-\alpha_1)/(\alpha_2-\alpha_1)$. 
  Density functions are not scaled to integrate to 1.}
  \label{table:kerneldensity}
\end{table}

We next list the backtests to be implemented, assigning to each a
unique mnemomic.
\begin{description}
\item[Binomial score test:] the two-sided binomial score test at level $\alpha^*$ (BIN);
\item[Multinomial tests:]  the 3-point discrete uniform kernel (ZU3) and 
                  3-point Pearson test (PE3);
\item[Continuous spectral tests:] univariate tests based on the uniform kernel (ZU); the arcsin kernel (ZA); Epanechnikov kernel (ZE);  and increasing (ZL$_+$) and decreasing (ZL$_{-}$) linear kernels;
\item[Continuous/mixed bispectral tests:] we
  combine the increasing and decreasing linear kernels (ZLL) and we also apply the  truncated
  probitnormal score test  (PNS).
\end{description}

We consider three different choices for the cdf $\losscdf$ of the \textit{true model} of $L_t$: the standard normal, the scaled $t_5$ and scaled  $t_3$. The Student $t$ distributions are scaled to have variance one so differences stem from different tail shapes rather than different variances.  We take the \textit{risk manager's model}  $\modelcdf$ to be the standard normal, i.e., we transform the sampled $L_t$ to PIT-values as $P_t=\Phi(L_t)$.  Therefore, when the samples of $L_t$ are drawn from the standard normal, the PIT-values are uniformly distributed and are used to evaluate
the size of the tests.  The PIT samples arising from the Student $t$
distributions show the kind of departures from uniformity that are
observed when the risk manager's model is too thin-tailed.  

We fix a sample size $n=750$ corresponding approximately to the three-year
samples of bank data studied in Section \ref{sec:empirical}.  In Table \ref{table:unconditionalZ},
we report the percentage of rejections of the null hypothesis at the 5\% confidence level based 
on $2^{16}=\num{65536}$ replications.  All reported $p$-values are
based on two-sided tests, though one-sided versions of some tests are
of course available.  

\begin{table}[htbp]
  \centering
  \begin{tabular}{*{2}{l}*{10}{r}}
    \toprule
    window & \( F \) \textbar\ kernel & \multicolumn{1}{c}{BIN} & \multicolumn{1}{c}{ZU3} & \multicolumn{1}{c}{PE3} & \multicolumn{1}{c}{ZU} & \multicolumn{1}{c}{ZA} & \multicolumn{1}{c}{ZE} & \multicolumn{1}{c}{\ZLp} & \multicolumn{1}{c}{\ZLn} & \multicolumn{1}{c}{ZLL} & \multicolumn{1}{c}{PNS} \\
    \midrule
    narrow & Normal & 6.1 & 4.9 & 5.3 & 4.7 & 4.7 & 4.7 & 4.6 & 4.8 & 4.8 & 4.9 \\
    & Scaled t5 & 33.9 & 35.0 & 40.3 & 33.8 & 34.4 & 33.0 & 40.3 & 27.1 & 40.0 & 44.7 \\
    & Scaled t3 & 24.0 & 24.8 & 43.4 & 23.9 & 24.3 & 23.3 & 32.7 & 16.5 & 43.3 & 50.5 \\ \addlinespace[3pt]
    wide & Normal & 6.1 & 5.0 & 5.1 & 4.9 & 4.9 & 4.9 & 4.9 & 4.9 & 5.0 & 5.0 \\
    & Scaled t5 & 33.9 & 10.7 & 55.5 & 6.4 & 6.6 & 6.1 & 11.9 & 5.8 & 45.1 & 57.5 \\
    & Scaled t3 & 24.0 & 13.5 & 90.6 & 17.7 & 20.4 & 15.4 & 7.4 & 31.9 & 85.8 & 93.1 \\
    \bottomrule
  \end{tabular}
  \caption{Estimated size and power of unconditional Z-tests.\\  We report the percentage of rejections of the null hypothesis at the 5\% confidence level based on $2^{16}=\num{65536}$ replications. The number of days in each backtest sample is $n=750$. The narrow window is [0.985, 0.995] and the wide window is [0.95, 0.995].}
  \label{table:unconditionalZ}
\end{table}

In both narrow and wide windows, we observe that the size of the Z-tests is very close to the nominal size of
5\%, except in the case of the binomial score test, which is slightly
oversized.  The power of the tests, in contrast, is sensitive to the choice of kernel. We summarize the results as follows:
\begin{enumerate}
\item Differences across tests in power are more pronounced on the
  wide window than on the narrow window. The monospectral tests (BIN, ZU, ZA, ZE, \ZLp, \ZLn)  are broadly similar in power to the binomial score test on the narrow window.
\item The monospectral tests offer more power against the scaled $t_5$ model than the more fat-tailed scaled $t_3$ model on the narrow window, but the opposite is true in most cases on the wide window.
\item Increasing the window width \textit{reduces} the power of most of the monospectral tests. 
\item For the wide window, the increasing linear kernel \ZLp\ offers more power than the decreasing linear kernel \ZLn when the true model is the scaled $t_5$, but the opposite holds when the true model is the scaled $t_3$.
\item The multispectral (PE3, ZLL, PNS) tests offer more power than the monospectral tests, but the differences are relatively small in the case of the narrow window when the true model is the scaled $t_5$.
\end{enumerate}

The first finding is easily understood. For the families of kernel densities in Table \ref{table:kerneldensity}, 
the associated function $G_\nu$ converges to a step function as the window narrows.  Put another way, all 
kernel families degenerate to the binomial score kernel 
as the window shrinks around $\alpha^*$.  To illuminate the second finding, we plot in Figure \ref{fig:dgpcdf} the distribution function for reported PIT-values under each of the true models, i.e.,
$\Pr(P_t\leq u)= \Pr\left(L_t\leq \Phi^{-1}(u)\right) = \losscdf\left(\Phi^{-1}(u)\right)$. 
The cdf is simply the identity line ($y=x$) when the null hypothesis
is true.  Within the narrow window of $[0.985, 0.995]$, the cdf for
the scaled $t_3$ lies closer to the identity line on average than does
the cdf for the scaled $t_5$. The monospectral tests, being 
  tests of the first moment of $W_t = G_\nu(P_t)$, are sensitive to
this distance, so have greater power against the scaled $t_5$ than the scaled $t_3$. On the wide window, however, the 
cdf for the scaled $t_5$ lies closer to the identity line on average, so the tests have greater power against the scaled $t_3$.

The figure also illuminutes the third and fourth findings.  Both of the scaled Student $t$ cdfs cross the identity line outside the boundaries of the narrow window, but near the middle of the wide window.  Crossings within the window reduce the average distance, so pose a particular challenge for the monospectral tests. On the wide window, the scaled $t_5$ cdf crosses the identity line near 0.971, which is slightly below the midpoint.  This slight asymmetry favors the \ZLp\ kernel, which puts heavier weight on the upper side of the window.  The scaled $t_3$ cdf crosses the identity line above the midpoint (near 0.982), and furthermore the distance between the cdf and identity line is much larger at the lower end of the window.  This asymmetry favors the \ZLn\ kernel.

Finally, the greater power of the multispectral tests is most apparent when the kernel window contains a crossing of the type just described.  While the crossing reduces the average distance between the cdf of the reported PIT-values and the identity line, the cdf will be too steep or too shallow.  Cross-moments of the kernels in a multispectral test can effectively detect such a \textit{slope violation}.  In contrast, when the cdf of the reported PIT-values lies roughly parallel to the identity line throughout the kernel window (as is the case for the scaled $t_5$ in the narrow window), the advantage of the multispectral test is expected to be less pronounced.

\begin{figure}[ht]
     \includegraphics[height=.46\textheight]{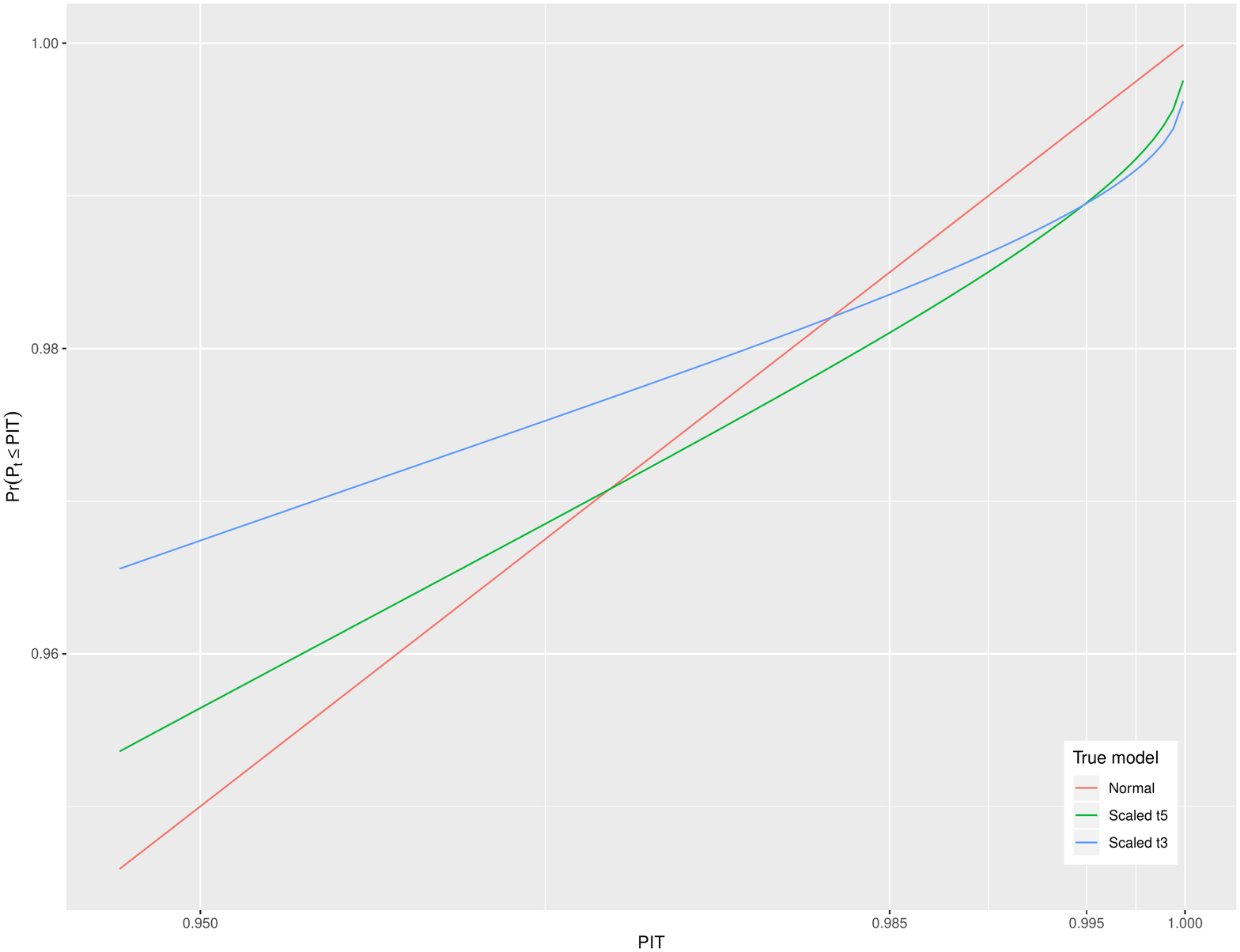}
     \caption{Distribution functions for reported PIT-values.\\
       CDFs for the reported PIT-values when the  risk manager assumes standard normal losses ($\modelcdf=\Phi$) but the true loss model  $\losscdf$ is standard normal (red line), scaled $t_5$ (green liine) or scaled $t_3$ (blue line).}
  \label{fig:dgpcdf}
\end{figure}

In our companion paper [\supplementary], we study three additional
dimensions of test design.  First within the set of discrete kernel Z-tests, we find that the 3-level tests considered in Table \ref{table:unconditionalZ} 
perform very similarly to their 5-level counterparts.  Second, our summary conclusions on Table \ref{table:unconditionalZ} are robust to the choice of backtest sample size $n$. Third, and
most importantly, we find that the multinomial Z-tests and the
probitnormal score test outperform their corresponding likelihood
ratio tests.   The Z-tests are similar in power to the LR-tests, but
the LR-tests are less well sized.\footnote{The superior performance of the Z-tests holds with a smaller backtest sample of $n=250$ as well.}   Therefore, we will henceforth limit our attention to the class of Z-tests.

As a general caveat, we do not advocate that power alone should
dictate the choice of kernel. A general
intuition from our simulation studies is that a test is most powerful
in rejecting a false model when the kernel weights heavily on
probability levels for which the quantiles of the risk manager's
model diverge from the true quantiles. As historical simulation in particular tends to understate the tails of the distribution, in practice we expect that the most powerful tests will weight heavily on extreme probability levels.  However, this can come at the expense of the stability of the test, in the sense that the outcome can be determined by the presence or absence of one or two very large reported 
PIT-values. Furthermore, testing at very extreme tail values of $\alpha$ runs counter to a primary regulatory motivation for the backtest, which is to verify the bank's 99\% VaR.  
% \subsection{A menagerie of tests and kernel functions}
% \label{sec:empirical:kernels}

\section{Tests of conditional coverage}\label{sec:tests_independence}
% Whereas unconditional tests are focused on testing for the
% correct distribution $F_W^0$ of the $W_t$ under the null hypothesis~(\ref{eq:nullhypothesis}),
% conditional backtests are joint tests of the correct distribution and the
% independence of $W_t$ and $\testfiltration_{t-1}$ for all $t$.
While the unconditional tests of
Section~\ref{sec:tests_unconditional} have some limited power to detect the
presence of serial dependencies, the aim in this section is to propose 
conditional extensions of our spectral tests that \emph{explicitly}
address the independence of $W_t$ and $\testfiltration_{t-1}$ \emph{as
  well as} the correctness of the distribution of $W_t$.
These tests should have more power to detect departures
from the null hypothesis resulting from a bank's failure to use all the
information in $\bankfiltration_{t-1}$ when building the predictive model
$\modelcdf_t$, such as a failure
 to address time-varying volatility in adequate fashion.

Our tests of conditional coverage extend the regression-based
approach to testing conditional coverage of VaR estimates and offer
many new possibilities. They include a variant on the 
widely-applied test of~\citet{bib:christoffersen-98}
and subsume the dynamic quantile (DQ) test of
\citet{bib:engle-manganelli-04}. The Christofferson test is an LR-test
of first-order Markov dependence in VaR exceedances and has been generalized to a
multilevel test by~\citet{Leccadito2014}. In the DQ test, VaR exceedance indicators are
regressed on lagged exceedance indicators and VaR  estimates
to assess whether exceedances occur independently at
the desired rate.

\subsection{Tests of the martingale difference property}
\label{sec:tests-indep-based}
A necessary condition for null hypothesis \eqref{eq:nullhypothesis} to hold is the martingale
difference (MD) property with respect to the regulator's filtration $(\testfiltration_t)$:
\begin{equation}
  \label{eq:36}
 E(W_t-\mu_W \mid \testfiltration_{t-1}) = 0
\end{equation}
where we recall that $\testfiltration_t=\sigma(\{P_s : s \leq t\})$.
When the MD property~\eqref{eq:36} holds, we must have
$E(h_{t-1}(W_t-\mu_W))=0$ for any 
$\testfiltration_{t-1}$-measurable random variable $h_{t-1}$. Using a
function $\lagpfunc$, which we refer to as a \textit{conditioning
variable transformation} (CVT), we
form the $k+1$-dimensional lagged vector 
$\bm{h}_{t-1}= (1,\lagpfunc(P_{t-1}),\ldots,\lagpfunc(P_{t-k}))^\prime$.
To guarantee the existence of the second moment
of $\bm{h}_{t-1}$, we assume that $(P_t)$ is covariance-stationary and
that  $\lagpfunc$ is bounded. Particular examples that 
we will use in our empirical analysis are $\lagpfunc(p) =\indicator{p \geq
  \alpha}$ for some $\alpha$ and $\lagpfunc(p) =  |2p -1|^\hpower$ for
$\hpower>0$. 

%To describe the tests we introduce the notation $(\zeromean{W}_t)$ for the sequence of
%transformed reported PIT-values $\zeromean{W}_t = W_t-\mu_W$ centered
%at their theoretical mean $\mu_W$ under the null
%hypothesis~\eqref{eq:nullhypothesis}.
%Recall that $(\testfiltration_t)$ is the regulator's filtration defined by 
%$\testfiltration_t=\sigma(\{P_s : s \leq t\})$.
%We test that $(\zeromean{W}_t)$ has the martingale
%difference (MD) property with respect to $(\testfiltration_t)$, which is necessary for~\eqref{eq:nullhypothesis} to hold:
%\begin{equation}
%  \label{eq:36}
%  \quad E(\zeromean{W}_t \mid \testfiltration_{t-1}) =0
%\end{equation}

For convenience, let $(\zeromean{W}_t)$ denote the sequence of
transformed reported PIT-values $\zeromean{W}_t = W_t-\mu_W$ centered
at their theoretical mean $\mu_W$.   
We base our test on the vector-valued process
$\bm{Y}_t = \bm{h}_{t-1}\zeromean{W}_t $ for $t= k+1,\ldots,n$.  Under
the null hypothesis~\eqref{eq:nullhypothesis}, $(\bm{Y}_t)$ is a MD
sequence satisfying $\E(\bm{Y}_t\mid\testfiltration_{t-1}) =\bm{0}$.
We want to test that
$\bm{Y}_{k+1},\ldots,\bm{Y}_n$ are close to the zero vector on
average.  We apply the conditional predictive test of~\citet{bib:giacomini-white-06}
which was developed for comparing forecasting methods, and which has
also been used by~\citet{bib:nolde-ziegel-17} in the backtesting
context.  Let
$\average{\bm{Y}}_{n,k} = (n-k)^{-1}\sum_{t=k+1}^n \bm{Y}_t$ and let
$\hat{\Sigma}_Y$ denote a consistent estimator of
$\Sigma_{Y}:= \cov(\bm{Y}_t)$.  Giacomini and White show that under very
weak assumptions, for large enough $n$ and fixed $k$,
\begin{equation}\label{eq:32}
 (n-k)\; \average{\bm{Y}}_{n,k}^\prime \; \hat{\Sigma}_{Y}^{-1} \; 
\average{\bm{Y}}_{n,k}
  \sim \chi^2_{k+1}.
\end{equation}
\citet{bib:giacomini-white-06} use the estimator $\hat{\Sigma}^{\scriptscriptstyle GW}_Y = (n-k)^{-1}\sum_{t=k+1}^n \bm{Y}_t \bm{Y}_t^\prime$
but we can use the fact that $\E(\zeromean{W}_t^2 \mid \testfiltration_{t-1})=\sigma_W^2$ for all $t$ under the null
hypothesis~\eqref{eq:nullhypothesis} to form an alternative estimator.
We compute that
\begin{align}
  \Sigma_Y = \E(\cov(\bm{Y}_t \mid \testfiltration_{t-1})) &= \E \left(\E
  \left( \bm{Y}_t \bm{Y}_t^\prime \mid \testfiltration_{t-1}\right)\right) \nonumber \\
  &= \E\left(\bm{h}_{t-1}\bm{h}_{t-1}^\prime\E\left(\zeromean{W}_t^2 \mid                                                                         
              \testfiltration_{t-1}\right)\right) = \sigma^2_W \lagpfuncmat
          \label{eq:SigmaYnull}
\end{align}
where $\lagpfuncmat = \E\left(\bm{h}_{t-1}\bm{h}_{t-1}^\prime\right)$,
which suggests the estimator $\hat{\Sigma}_Y = \sigma_W^2
\hat{\lagpfuncmat}$ where\footnote{It would also be possible to use a heteroscedasticity
  and autocorrelation-consistent (HAC) estimator of $\Sigma_Y$~\citep{bib:newey-west-87,bib:andrews-91}. 
  Our preferred estimator has the advantage of higher sensitivity to deviations from the null hypothesis of zero
 serial correlation.}
\begin{equation}
  \label{eq:47}
  \hat{\lagpfuncmat} =  (n-k)^{-1}\sum_{t=k+1}^n \bm{h}_{t-1}\bm{h}_{t-1}^\prime.
\end{equation}

The decomposition in~\eqref{eq:SigmaYnull} has the advantage that it
generalizes our unconditional spectral Z-test, which corresponds to the case $k=0$. 
The case $k=1$ may be viewed as a Z-test version of the first-order
Markov chain test of~\citet{bib:christoffersen-98}.  To see that the
conditional test also embeds the DQ test
statistic proposed by~\citet{bib:engle-manganelli-04}, let 
$X$ be the $(n-k)\times(k+1)$ matrix whose rows are given by
 $\bm{h}_{t-1}$ for $t=k+1,\ldots,n$ and
let $\bm{\zeromean{W}}=(\zeromean{W}_{k+1},\ldots,\zeromean{W}_n)^\prime$. 
It follows that 
\[ \hat{\Sigma}_Y = \sigma_W^2 (n-k)^{-1}\sum_{t=k+1}^n 
\bm{h}_{t-1}
\bm{h}_{t-1}^\prime = \sigma_W^2 (n-k)^{-1} X^\prime X \] and
$\average{\bm{Y}}_{n,k} = (n-k)^{-1} X^\prime \bm{\zeromean{W}}$ so 
that~\eqref{eq:32} may be rewritten as
\begin{equation}
  \label{eq:35}
  \sigma_W^{-2}\bm{\zeromean{W}}^\prime X (X^\prime X)^{-1} X^\prime 
\bm{\zeromean{W}}
  \sim \chi^2_{k+1}.
\end{equation}
The DQ test corresponds
to the binomial score case, i.e., the case where
$W_t= \indicator{P_t \geq \alpha}$ and 
the CVT is $\lagpfunc(p) =\indicator{p \geq
  \alpha}$\footnote{\citet{bib:engle-manganelli-04} allow as well  for
  lagged VaR values to be included as regressors, an extension
  possible in our framework, but change in
  portfolio composition implies that lagged VaR values are less
  informative than lagged PIT values.}'
 
 It is straightforward to generalize the conditional spectral Z-test to a conditional bispectral Z-test; see
 Appendix \ref{app:cond-bisp-test}.

\subsection{Size and power}
\label{sec:mcConditional}

We build on the Monte Carlo exercises of Section
\ref{sec:mcUnconditional} to study the size and power of the
conditional tests of coverage. Here we present a representative
extract of simulation studies documented in our companion paper [\supplementary].  
The data are generated from three different ``true'' models: iid standard normal; a
time series model with underlying ARMA(1,1) structure and standard normal marginals; and a
model with underlying ARMA(1,1) structure and scaled $t_5$ marginal distribution.  Our calibration of the
ARMA parameters (AR = 0.95, MA = -0.85) is described in the companion
paper [\supplementary] and is designed to mimic the serial dependence
in PIT values when stochastic volatility is neglected.   
As in Section \ref{sec:mcUnconditional}, we assume the risk manager reports PIT-values based on the
standard normal model $\modelcdf=\Phi$. 

%\noteA[inline]{Could cite separate paper on VT-ARMA models (currently
%  unsubmitted draft). The text
%  in supplement is somewhat compressed and may be difficult to follow.}
%\noteM[inline]{This would be discouraged in a double-blind review.  In any case, it seems premature to cite that
% draft. Very sorry, by the way, that my promised comments are overdue.}

In addition to a choice of kernel, the MD test requires the choice of the number ($k$) of lagged PIT values and the conditioning variable transformation $\lagpfunc(P)$.    Define $V(u)=\vert 2u-1\vert$; this V-shaped transformation of PIT values is well-suited to uncover dependence arising from stochastic volatility.   As listed in Table \ref{table:cvt}, we consider four candidates for the CVT.  Whereas the \EMone\ requires only a time-series of traditional exceedance indicators, the three CVT based on the $V(u)$ transformation require that the regulator observe PIT values. 

\begin{table}[htbp]
  \centering
  \begin{tabular}{lll}
    \toprule
   Mnemonic & $\lagpfunc(P)$ & Description \\
    \midrule
  \EMone & $\indicator{P\geq 0.99}$ & Flags upper-tail PIT values, as in \citet{bib:engle-manganelli-04}.\\
  \EMtwo & $\indicator{V(P)\geq 0.98}$ &  Two-tailed version of \EMone, flags PIT values near zero or one.  \\
  \MDfour & $V(P)^4$ & Places heavier weight on tail PIT values in the recent past.\\
  \MDhalf & $\sqrt{V(P)}$ & Dampens sensitivity to tail PIT values relative to \MDfour.\\
    \bottomrule
  \end{tabular}
  \caption{Conditioning variable transformations. $V(u)\equiv \vert 2u-1\vert$}
  \label{table:cvt}
\end{table}

Table~\ref{table:extracted_table} gives a flavor of the
main findings for the example of the uniform kernel (ZU) and the narrow kernel window of 
$[0.985, 0.995]$. We report the percentage of rejections of the null
hypothesis at the 5\% confidence level based on $2^{16}=\num{65536}$
replications.  In the first column (CVT=``None''), we set $k=0$ to
obtain the unconditional Z-test.  Each of the remaining columns
corresponds to a CVT with $k=4$.   As seen in the first row (iid
standard normal model), size is more difficult to control in the
conditional tests. However, the CVT choices \MDfour\ and \MDhalf\ are only
slightly oversized whereas \EMtwo\ and, in particular, \EMone\ are very
oversized.
%\noteA[inline]{Should we add results for $k=1$ in companion paper? One
%  referee asked for that and it would allow us to claim superiority to
%  both the
%  Engle+Manganelli and Christoffersen tests.}

The model depicted in the second row gives uniformly
distributed PIT values with a serial dependence structure that is
typical when stochastic volatility is ignored. The power of the unconditional test in this situation is very limited (10.8\%),
while the MD tests show power ranging from 21.7\% to 32.6\%. There is a further increase in power 
when the simulated PIT data
are both non-uniform and serially dependent (third row). 

\begin{table}[ht]
\centering
\begin{tabular}{lllllll}
  \toprule
F & Serial Dependence $|$ CVT & None & \EMone & \EMtwo & \MDfour & \MDhalf \\ 
  \midrule
  Normal &             None &  4.8 & 14.4 &  9.0 &  6.7 &  6.7 \\ 
  Normal & ARMA(1, 1) & 10.8 & 31.5 & 30.9 & 32.6 & 21.7 \\ 
  Scaled t5 & ARMA(1, 1) & 36.2 & 54.9 & 52.7 & 60.7 & 54.5 \\ 
   \bottomrule
\end{tabular}
\caption{Estimated size and power of conditional tests.\\  MD tests using the ZU kernel on the narrow window [0.985, 0.995]. We report the percentage of rejections of the null hypothesis at the 5\% confidence level based on $2^{16}=\num{65536}$ replications. The number of days in each backtest sample is $n=750$. ARMA parameters are AR = 0.95, MA = -0.85.} 
\label{table:extracted_table}
\end{table}

\section{Application to bank-reported PIT values}
\label{sec:empirical}

\subsection{Data} \label{sec:empirical:data}

Our data consist of ten confidential backtesting samples provided by US banks to the Federal Reserve Board at the subportfolio level.  Mandatory reporting to bank regulators pursuant to the Market Risk Rule  took effect on January 1, 2013.   For each significant subportfolio and each business day, the bank is required to report the overnight VaR at the 99\% level, the realized clean \PnL, and the associated PIT-value \citep[p.~53105]{bib:rbcg:marketrisk}.  While the first two fields have been available to regulators for a long time (at least at an aggregate trading book level), access to PIT values is new.
Each of our ten samples represents returns on an equity or foreign
exchange subportfolio, which can include derivative as well as cash positions.  
Our samples are taken from the three-year period from 2014--2016. 

Summary statistics for the unconditional distributions are found in Table \ref{tab:summarystats}.   As is often the case with new regulatory reporting requirements, the data are not uniform in quality.  Two of the samples (coded \portfolio104 and \portfolio110) have missing values (0.9\% and 3.2\% of trading days, respectively).  Furthermore, close inspection reveals that most of the samples contain a small number of observations that are potentially  spurious. In a few extreme cases, a PIT value of 1 is matched to a realized loss smaller than the forecast VaR. We apply a heuristic procedure to identify spurious values based on the distance between the reported PIT-value and an imputed value.  The latter is constructed
using a portfolio-specific model that fits PIT to the ratio of realized loss to VaR; details are provided in Appendix \ref{app:spurious}.  In test results reported below, we treat spurious values as missing to make the tests less sensitive to 
reporting error.  Our conclusions are robust to taking all non-missing observations as valid.

Remaining columns of the table provide a histogram of PIT values.  For some portfolios, tail PIT values 
are underrepresented (e.g., \portfolio107) or overrepresented  (e.g.,  \portfolio105) in the sample.  For some other portfolios, the histograms appear to be close to uniform, e.g., for \portfolio110, 85.9\% of PIT values lie in $[0.05, 0.95)$ and 
remaining mass is distributed roughly symmetrically.  

\begin{sidewaystable}
\centering
\begin{tabular}{r rrr rrrrrrr}
 \toprule
\multicolumn{1}{c}{ID} & \multicolumn{1}{r}{Trading days} &
 \multicolumn{2}{c}{of which:} & \multicolumn{7}{c}{Frequencies} \\ 
  \cmidrule{3-4} 
  &  & Missing & Spurious & $[0,.005)$ & $[.005,.015)$ & $[.015,.05)$ & $[.05,.95)$ & $[.95,.985)$ & $[.985,.995)$ & $[.995,1]$ \\ 
  \midrule
  101 &   756 &  0 &  0 & 0.0132 & 0.0172 & 0.0357 & 0.8796 & 0.0317 & 0.0106 & 0.0119 \\ 
   102 &   751 &  0 &  7 & 0.0027 & 0.0108 & 0.0215 & 0.9113 & 0.0323 & 0.0121 & 0.0094 \\ 
   103 &   750 &  0 &  8 & 0.0040 & 0.0040 & 0.0189 & 0.9474 & 0.0162 & 0.0081 & 0.0013 \\ 
   104 &   774 &  7 &  0 & 0.0000 & 0.0000 & 0.0026 & 0.9804 & 0.0104 & 0.0013 & 0.0052 \\ 
   105 &   750 &  0 &  2 & 0.0174 & 0.0214 & 0.0294 & 0.8596 & 0.0388 & 0.0187 & 0.0147 \\ 
   106 &   629 &  0 &  1 & 0.0111 & 0.0191 & 0.0510 & 0.8328 & 0.0557 & 0.0175 & 0.0127 \\ 
   107 &   750 &  0 &  1 & 0.0000 & 0.0000 & 0.0013 & 0.9960 & 0.0027 & 0.0000 & 0.0000 \\ 
   108 &   756 &  0 &  8 & 0.0000 & 0.0053 & 0.0267 & 0.9278 & 0.0174 & 0.0120 & 0.0107 \\ 
   109 &   734 &  0 &  6 & 0.0082 & 0.0151 & 0.0495 & 0.8654 & 0.0371 & 0.0206 & 0.0041 \\ 
   110 &   774 & 25 &  0 & 0.0134 & 0.0200 & 0.0507 & 0.8585 & 0.0427 & 0.0107 & 0.0040 \\ 
   \bottomrule
\end{tabular}
\caption{Sample statistics.\\ Missing and spurious observations excluded from the reported frequencies. Sample period is 2014-01-01 to 2016-12-31.} 
\label{tab:summarystats}
\end{sidewaystable}

\subsection{Tests of unconditional coverage}

Due to the generality of our framework, application of spectral backtests to data 
involves choices along several dimensions.  As 
in Section \ref{sec:mcUnconditional}, we fix $\alpha^*=0.99$ as the 
conventional VaR level, define a \textit{narrow} window as $[0.985, 0.995]$ and
a \textit{wide} window as $[0.95, 0.995]$. Kernels are drawn from
Table \ref{table:kerneldensity}.  Guided by our simulation results and the need
for brevity, we exclusively employ two-sided Z-tests in our empirical analysis.

Table \ref{tab:unconditional} presents $p$-values for the tests of unconditional coverage.\footnote{All $p$-values in the tables below should be interpreted in the context of a single test of the null
  hypothesis. If multiple tests are conducted, inferences would have to
  be based on a standard correction method such as that of Bonferroni;
  see~\citet{bib:schaffer-95} for a review.}  We find that the forecast models for portfolios \portfolio105, \portfolio106 and \portfolio107
are rejected at the 1\% level for all kernels and on both the narrow and wide kernel windows.  In view of the histograms observed in Table \ref{tab:summarystats}, this is unsurprising.   
When an empirical distribution function (edf) lies above the uniform cdf within the kernel window (as observed for \portfolio107), large PIT values are underepresented in the sample, which suggests that the forecast model overstates the upper quantiles of the loss distribution.  When an edf lies below the uniform cdf (as observed for \portfolio105 and \portfolio106), large PIT values are overrepresented in the sample, which suggests that the forecast model understates the upper quantiles.  By contrast, there are no rejections at all for \portfolio110, for which the edf is reasonably close to the theoretical cdf throughout the upper tail.

\begin{sidewaystable}
\centering
\begin{tabular}{rlrrrrrrrrrr}
  \toprule
 ID & window & BIN & ZU3 & PE3  & ZU & ZA & ZE & \ZLp & \ZLn & ZLL &  PNS \\ 
  \midrule
  \multirow{2}{*}{101} & narrow & 0.1046 & 0.0340 & 0.0502 & 0.0356 & 0.0294 & 0.0436 & 0.0230 & 0.0572 & 0.0595 & 0.0352 \\ 
   & wide & 0.1046 & 0.1407 & 0.0593 & 0.2050 & 0.2300 & 0.1718 & 0.0610 & 0.4243 & 0.0128 & 0.0296 \\ \midrule
  \multirow{2}{*}{102} & narrow & 0.0016 & 0.0200 & 0.0027 & 0.0800 & 0.0795 & 0.0802 & 0.1346 & 0.0558 & 0.1085 & 0.2153 \\ 
   & wide & 0.0016 & 0.0063 & 0.0121 & 0.2456 & 0.2484 & 0.2797 & 0.2114 & 0.2898 & 0.4477 & 0.2238 \\ \midrule
  \multirow{2}{*}{103} & narrow & 0.1029 & 0.1158 & 0.4212 & 0.0987 & 0.1134 & 0.0882 & 0.1094 & 0.0995 & 0.2545 & 0.2804 \\ 
   & wide & 0.1029 & 0.0035 & 0.0229 & 0.0058 & 0.0041 & 0.0083 & 0.0156 & 0.0038 & 0.0126 & 0.0091 \\ \midrule
  \multirow{2}{*}{104} & narrow & 0.1829 & 0.1691 & 0.1119 & 0.1651 & 0.1678 & 0.1633 & 0.3063 & 0.0987 & 0.0767 & 0.0657 \\ 
   & wide & 0.1829 & 0.0010 & 0.0002 & 0.0004 & 0.0003 & 0.0005 & 0.0031 & 0.0002 & 0.0003 & 0.0001 \\ \midrule
  \multirow{2}{*}{105} & narrow & 0.0000 & 0.0000 & 0.0001 & 0.0000 & 0.0000 & 0.0000 & 0.0000 & 0.0000 & 0.0000 & 0.0000 \\ 
   & wide & 0.0000 & 0.0000 & 0.0002 & 0.0000 & 0.0000 & 0.0000 & 0.0000 & 0.0001 & 0.0001 & 0.0000 \\ \midrule
  \multirow{2}{*}{106} & narrow & 0.0005 & 0.0004 & 0.0055 & 0.0010 & 0.0008 & 0.0012 & 0.0019 & 0.0007 & 0.0033 & 0.0036 \\ 
   & wide & 0.0005 & 0.0000 & 0.0001 & 0.0000 & 0.0000 & 0.0000 & 0.0000 & 0.0000 & 0.0000 & 0.0000 \\ \midrule
  \multirow{2}{*}{107} & narrow & 0.0059 & 0.0018 & 0.0097 & 0.0026 & 0.0020 & 0.0032 & 0.0062 & 0.0016 & 0.0053 & 0.0033 \\ 
   & wide & 0.0059 & 0.0000 & 0.0000 & 0.0000 & 0.0000 & 0.0000 & 0.0000 & 0.0000 & 0.0000 & 0.0000 \\ \midrule
  \multirow{2}{*}{108} & narrow & 0.0166 & 0.0213 & 0.0933 & 0.0311 & 0.0290 & 0.0301 & 0.0191 & 0.0520 & 0.0482 & 0.0635 \\ 
   & wide & 0.0166 & 0.6971 & 0.0052 & 0.8422 & 0.8168 & 0.8545 & 0.5358 & 0.4445 & 0.0034 & 0.0070 \\ \midrule
  \multirow{2}{*}{109} & narrow & 0.5217 & 0.2499 & 0.0222 & 0.1358 & 0.1237 & 0.1712 & 0.3171 & 0.0634 & 0.0165 & 0.0225 \\ 
   & wide & 0.5217 & 0.2584 & 0.3373 & 0.0563 & 0.0715 & 0.0525 & 0.0503 & 0.0706 & 0.1471 & 0.2814 \\ \midrule
  \multirow{2}{*}{110} & narrow & 0.8514 & 0.9479 & 0.8692 & 0.6660 & 0.7719 & 0.6123 & 0.6321 & 0.7032 & 0.8737 & 0.9856 \\ 
   & wide & 0.8514 & 0.5407 & 0.6916 & 0.3017 & 0.2932 & 0.3442 & 0.4527 & 0.2325 & 0.3407 & 0.6253 \\ \midrule
\bottomrule
\end{tabular}
\caption{Tests of unconditional coverage.\\ We report $p$-values by portfolio, kernel window, and kernel family. Narrow kernel window is [0.985,0.995] and wide kernel window is [0.95,0.995]. Sample period is 2014-01-01 to 2016-12-31.} 
\label{tab:unconditional}
%\end{threeparttable}
\end{sidewaystable}

For the remaining six portfolios, test results are sensitive to the choice of kernel.  
This is to be expected and desirable, as the different tests prioritize different quantiles of the unconditional distribution. To shed light on the differences, in Figure \ref{fig:empcdf} we plot the edf for three portfolios on the narrow window (upper panel) and wide window (lower panel).  This plot is the empirical counterpart to Figure \ref{fig:dgpcdf} in Section \ref{sec:mcUnconditional}. For \portfolio104, we observe that the edf intersects with the theoretical cdf at the common upper window boundary $\alpha_2=0.995$, but lies above the theoretical cdf at lower PIT values.  Over the wide window, the average distance between edf and theoretical cdf is large, so any test that assigns significant weight to PIT values near the center of the window will reject.   When restricted to the narrow window, the average distance is reduced, so the tests fail to reject. The edf for portfolio \portfolio103 (not shown) is qualitatively close to that of \portfolio104, which explains the similarity in test results.

The edf for portfolio \portfolio108 lies somewhat below and roughly parallel to the theoretical cdf throughout the narrow window.  Tests reject at the 5\% level for some of the kernels and fail to reject for others, but the $p$-values all lie between 1.5\% and 10\%.  When we consider the wide window, we find that the edf lies above the theoretical cdf in the lower half of the window and below in the upper half, which implies that  the forecast model 
underestimates quantiles at one boundary of the kernel window and overestimates quantiles at the other boundary, i.e., a slope deviation 
from the uniform cdf.   As shown in Section \ref{sec:mcUnconditional}, bispectral tests generally outperform monospectral tests in this situation.   We find that the tests based on the bivariate ZLL and PNS kernels and the trivariate Pearson kernel reject at the 1\% level.  The edf for portfolio \portfolio101 (not shown) also displays a slope violation on the wide window, and again we find the bivariate ZLL and PNS kernels most effective. 

In the case of portfolio \portfolio109, the edf displays a slope violation within the narrow window.  As before, we find that the tests based on the bivariate ZLL and PNS kernels and the trivariate Pearson kernel reject at the 5\% level, whereas the monospectral tests all fail to reject.  On the wide window, the edf lies uniformly below the theoretical cdf, so the slope violation loses salience.  The bispectral and trispectral tests now fail to reject, whereas several of the monospectral tests reject at a 10\% level.

\begin{figure}[!htbp]
     \includegraphics[height=.46\textheight]{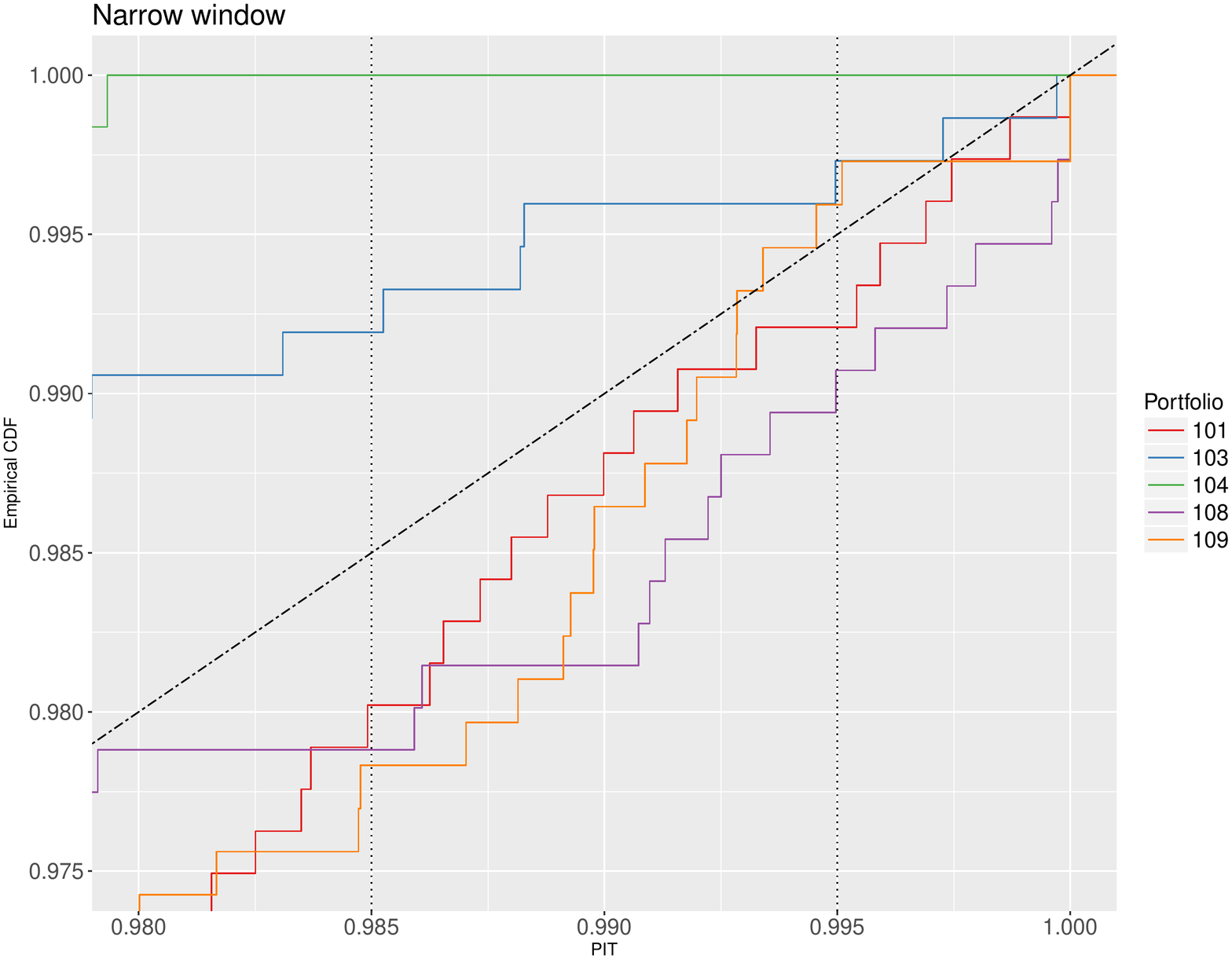}

     \includegraphics[height=.46\textheight]{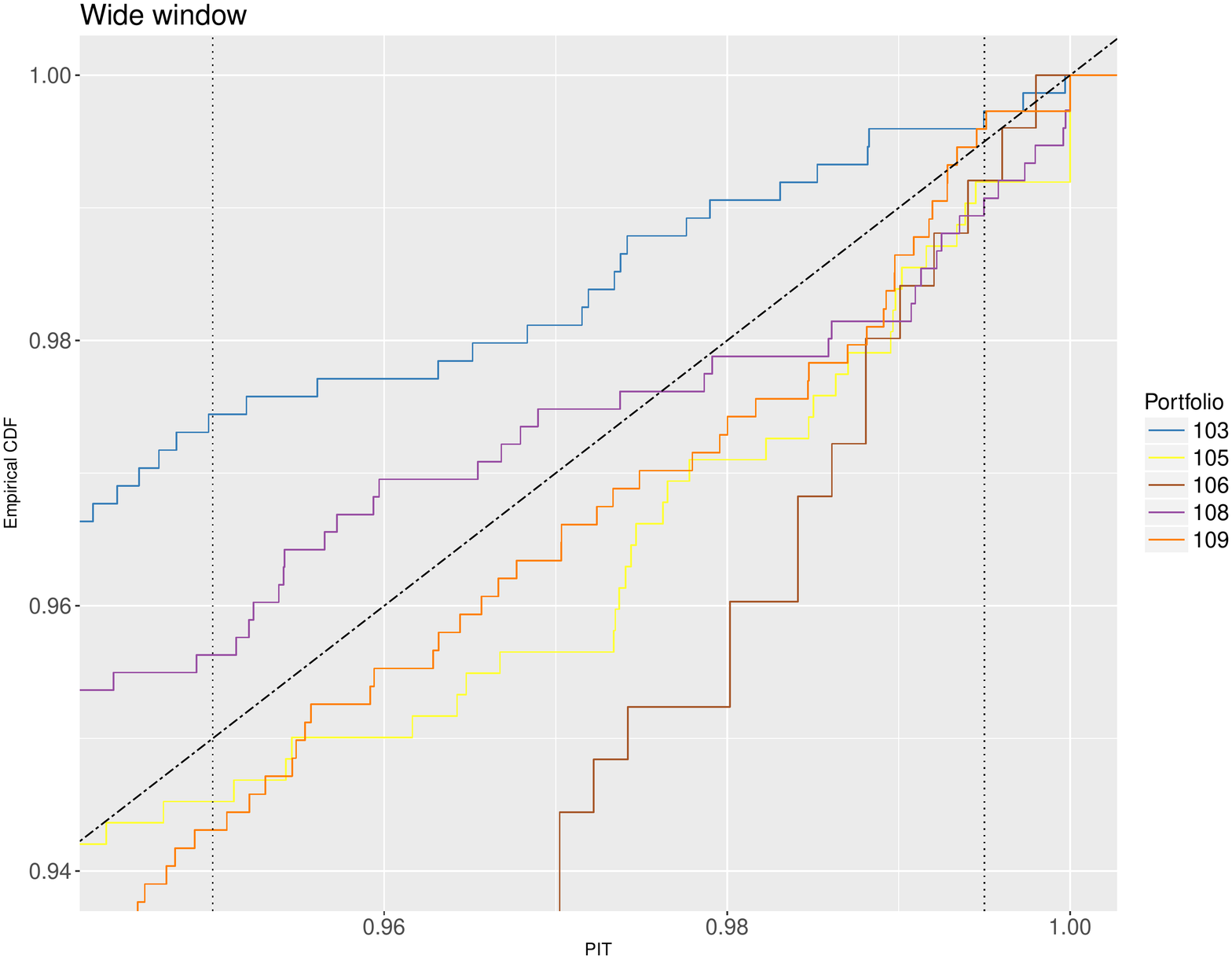}

  \caption{Empirical distribution functions for select portfolios.\\
       EDFs for narrow kernel window (upper panel) and wide kernel window (lower panel).  The uniform cdf is plotted as a dashed black line.}
  \label{fig:empcdf}
\end{figure}

\subsection{Tests of conditional coverage}
\label{sec:emp:conditional}

In this section, we emphasize the role of the conditioning variable transformation $\lagpfunc(P)$ in revealing serial dependence in PIT-values.
For parsimony, we consider only a subset of the kernels used in the previous section.  We include the binomial score kernel (BIN) as representative of the traditional test, the uniform kernel (ZU) as representative of the continuous monospectral tests, and the (ZLL) as representative of the multispectral tests. We fix $k=4$ lags in the monospectral tests which corresponds to
looking at dependencies over a time horizon of one trading
week.  To facilitate comparison to the monospectral tests,  we fix $(k_1=4, k_2=0)$ for the bispectral ZLL test. 

Missing or spurious values may be especially troublesome in a test of conditional coverage because a  PIT value missing at time $t$ introduces missing regressors at $t+1,\ldots, t+k$.  To avoid losing the subsequent $k$ observations, we replace missing or spurious $P_{t-\ell}$ with an inputed value when computing the lagged vector $\bm{h}_{t-1}$.  (As in the tests of unconditional coverage, we do not impute missing $P_t$ to backfill the dependent variables $W_t$, but simply drop these observations.)  Details of our imputation algorithm are found in Appendix \ref{app:spurious}.

\begin{table}
\centering
\begin{tabular}{rlrrrrr}
  \toprule
     &        &      & \multicolumn{2}{c}{narrow window} & \multicolumn{2}{c}{wide window} \\
     \cline{4-5} \cline{6-7} 
ID & CVT & BIN & ZU & ZLL & ZU & ZLL \\ 
  \midrule
    \multirow{4}{*}{101} & \EMone & 0.0017 & 0.0002 & 0.0020 & 0.0342 & 0.0092 \\ 
   & \EMtwo & 0.0084 & 0.0002 & 0.0009 & 0.0513 & 0.0115 \\ 
   & \MDfour & 0.0062 & 0.0003 & 0.0007 & 0.0921 & 0.0230 \\ 
   & \MDhalf & 0.0658 & 0.0063 & 0.0078 & 0.1376 & 0.0301 \\ \midrule
    \multirow{4}{*}{102} & \EMone & 0.0088 & 0.0975 & 0.1448 & 0.1873 & 0.1443 \\ 
   & \EMtwo & 0.0119 & 0.4716 & 0.3441 & 0.0201 & 0.0147 \\ 
   & \MDfour & 0.0000 & 0.0086 & 0.0023 & 0.0003 & 0.0006 \\ 
   & \MDhalf & 0.0000 & 0.0239 & 0.0111 & 0.0004 & 0.0005 \\ \midrule
    \multirow{4}{*}{103} & \EMone & 0.7540 & 0.7429 & 0.8408 & 0.1787 & 0.1879 \\ 
   & \EMtwo & 0.0773 & 0.0426 & 0.1331 & 0.0068 & 0.0149 \\ 
   & \MDfour & 0.3064 & 0.1372 & 0.2286 & 0.0074 & 0.0088 \\ 
   & \MDhalf & 0.5209 & 0.4194 & 0.5181 & 0.0318 & 0.0276 \\ \midrule
    \multirow{4}{*}{104} & \EMone & 0.8732 & 0.8501 & 0.5202 & 0.0295 & 0.0138 \\ 
   & \EMtwo & 0.8700 & 0.8456 & 0.5167 & 0.0289 & 0.0135 \\ 
   & \MDfour & 0.3225 & 0.2013 & 0.1287 & 0.0034 & 0.0022 \\ 
   & \MDhalf & 0.2085 & 0.0976 & 0.0689 & 0.0016 & 0.0011 \\ \midrule
    \multirow{4}{*}{107} & \EMone & NA & NA & NA & NA & NA  \\ 
   & \EMtwo & NA & NA & NA & NA & NA  \\ 
   & \MDfour & 0.1844 & 0.1071 & 0.1085 & 0.0001 & 0.0000 \\ 
   & \MDhalf & 0.1844 & 0.1071 & 0.1085 & 0.0001 & 0.0000 \\ \midrule
    \multirow{4}{*}{109} & \EMone & 0.9917 & 0.6351 & 0.0548 & 0.1465 & 0.1383 \\ 
   & \EMtwo & 0.8041 & 0.7510 & 0.1179 & 0.4522 & 0.5170 \\ 
   & \MDfour & 0.8929 & 0.8603 & 0.2067 & 0.5578 & 0.6225 \\ 
   & \MDhalf & 0.9313 & 0.6389 & 0.1327 & 0.4766 & 0.5266 \\ \midrule
    \multirow{4}{*}{110} & \EMone & 0.2658 & 0.3058 & 0.4121 & 0.1661 & 0.1395 \\ 
   & \EMtwo & 0.0041 & 0.0006 & 0.0009 & 0.0044 & 0.0108 \\ 
   & \MDfour & 0.0093 & 0.0008 & 0.0012 & 0.0100 & 0.0352 \\ 
   & \MDhalf & 0.1403 & 0.0513 & 0.0840 & 0.1508 & 0.2823 \\ 
\bottomrule
\end{tabular}
\caption{Tests of conditional coverage.\\ We report test $p$-values by portfolio, conditioning
variable transformation, kernel window and kernel family. The monospectral tests utilize $k=4$ lags, and for the ZLL bispectral test we set $(k_1=4, k_2=0)$.  Narrow kernel window is [0.985,0.995] and wide kernel window is [0.95,0.995]. Sample period is 2014-01-01 to 2016-12-31. Forecast models for \portfolio105, \portfolio106 and \portfolio108 (not tabulated) are rejected at the 1\% level for all choices of CVT and kernel.}
\label{tab:conditional}
\end{table}

Table \ref{tab:conditional} presents $p$-values for the tests of conditional coverage.  
For portfolios \portfolio105, \portfolio106 and \portfolio108, forecast models are strongly rejected (always at the 1\% level, and nearly always at the 0.01\% level) regardless of the choice of CVT or kernel; for brevity we drop these portfolios from the table.  
For only a single portfolio (\portfolio109), the forecast model is never rejected.  In the other six cases, choice of CVT and kernel matters. 

For portfolios \portfolio102  and  \portfolio110, the \MDfour\ CVT generally leads to rejection at the 5\% level, whereas tests using the \EMone\ CVT generally do not.  The \EMtwo\ and \MDhalf\ CVT are effective in many cases, but appear less robust than \MDfour.  This reflects the greater sensitivity of the \MDfour\ transformation to local spikes in market volatility.  In the case of portfolio \portfolio103, the tests reject on the wide window except when using the \EMone\ CVT. 

For portfolios \portfolio101 and \portfolio104, variation in $p$-value across tests is driven primarily by kernel choice, and in a manner consistent with the tests of unconditional coverage in Table \ref{tab:unconditional}.  Thus, for these two portfolios, serial dependence in the PIT-values does not appear to be the salient shortcoming in the forecast model.

In the case of  \portfolio107, the test statistic is undefined for the  \EMone\ CVT and its two-tailed counterpart (\EMtwo).  
As there were no observed violations in either tail ($P_t<.01$ or $P_t>.99$), in both cases the matrix 
$\hat{\lagpfuncmat}$ of \eqref{eq:47} is singular, so $\hat{\Sigma}_Y$ in the test statistic cannot be inverted.  This demonstrates a practical limitation of a binary-valued CVT, as short samples may often contain no tail values.  Observe also that the backtest fails to reject for the remaining two CVT on the narrow window, even though the forecast model for this portfolio is strongly rejected by the unconditional tests.  Since $P_t<\alpha_1=0.985$ for all $t$, $W_t$ has a degenerate distribution in the sample. In this situation, it may be shown that the conditional test statistic is invariant to the CVT and to $k$ and is equal to the unconditional test statistic.  Recalling that the test statistic has distribution $\chi^2_{1+k}$ under the null hypothesis, we find that the $p$-value increases with $k$. This explains why unconditional backtests may have greater power than conditional backtests in situations where an overly conservative forecast model leads to degeneracy in $W_t$.

\section{Conclusion}
The class of spectral backtests embeds many of the most widely used
tests of unconditional coverage and tests of conditional coverage,
including the binomial likelihood ratio test of \citet{bib:kupiec-95},
the interval likelihood ratio test of \citet{bib:berkowitz-01}, and
the dynamic quantile test of \citet{bib:engle-manganelli-04}. 
As we demonstrate with many examples, viewing these tests in terms
of the associated kernels facilitates the construction of
new tests.  From the perspective of the practice of risk
management, making explicit the choice of kernel measure may help
to discipline the backtesting process because the kernel 
directly expresses the user's priorities for model performance.

Different kernels are sensitive to different deviations from the null
hypothesis. A tester who only cares about systematic under- or
overestimation of quantiles within a narrow range is well served by a
number of single kernels, discrete and continuous. A tester who wants
to ensure maximum fidelity of the forecast models to the true
distributions across a wider range of quantiles may worry more about
slope violations (overestimation of quantiles at one end of a window and
underestimation at the other). Such a tester may favor
a multispectral test. However, to promote a single ``best'' test from the spectral family 
would be contrary to
the philosophy of our contribution, and we refrain from doing so. The tester should reflect on
performance priorities and select her kernel accordingly.

Our results illustrate the value to regulators of access to
bank-reported PIT-values.  Until recently, regulators effectively
observed only a sequence of VaR exceedance event indicators at a
single level $\alpha$, and therefore backtests were designed to take
such data as input.  In some jurisdictions, including the United
States, PIT-values have been collected for some time. Besides enabling the formation of spectral test 
statistics, lagged PIT-values are especially effective as
conditioning variables in regression-based tests of conditional
coverage.

%\clearpage
\appendix
\numberwithin{equation}{section}
 \section{Proofs}\label{sec:proofs}
\subsection{Proof of Theorem~\ref{prop:product-of-Ws}}\label{sec:moments-spectr-trans}
Let $G_1$ and $G_2$ be the increasing, right-continuous functions
associated with the measures $\nu_1$ and $\nu_2$. It follows that
$W_t^*= G_1(P_t)G_2(P_t)$. The function $G^*(u)=G_1(u)G_2(u)$ must
also be increasing and right-continuous and can thus be used to define
a Lebesgue-Stieltjes measure $\nu^*$ by setting $\nu^*(\{0\})=G^*(0)=0$ and
$\nu^*((a,b]) = G^*(b)-G^*(a)$ for any $0 \leq a < b \leq 1$. It
follows that $W_t^*=G^*(P_t) = \nu^*([0,P_t])$.

The formula for $\nu^*$ is obtained by applying the integration-by-parts formula for the 
Lebesgue-Stieltjes integral \citep[Theorem A]{bib:hewitt-60}.

\subsection{Proof of Proposition~\ref{prop:tailbehavior}}

Since $G_\nu(u) = \bigO((1-u)^{-1/2+\epsilon})$ as $u \to 1$ for some
small $\epsilon$, there exists a value $u_0$ and a positive constant
$C$ such that $G_\nu(u) \leq C (1-u)^{-1/2+\epsilon}$ for $u\geq
u_0$. Let $\bar{u}$ be the larger of $u_0$ and the last point at which
$G_\nu$ is not differentiable (there are only finitely so many points
by Assumption \ref{ass:almostcontinuous}). We can decompose \eqref{eq:mustar} as
\begin{equation}
\label{eq:mustarbroken}
  \E(W_t^2) = \int_{[0, \bar{u}]}(1-u)\left(2G_\nu(u)
    -\nu(\{u\})\right)\rd \nu(u) + \int_{(\bar{u}, 1]}(1-u)\left(2G_\nu(u)
    -\nu(\{u\})\right)\rd \nu(u)
\end{equation}
The integrand in the first term is bounded above by $2G_\nu(\bar{u})$ and
so the integral is finite. We only need to prove the finiteness of the second term
which can be written as
\begin{align*}
\int_{\bar{u}}^1(1-u)2G_\nu(u) g_\nu(u) \rd u &= \int_{\bar{u}}^1(1-u)
                                                   \frac{\rd}{\rd u}
                                                   \left(G_\nu(u)^2
                                                   \right) \rd u = \left[ G_\nu(u)^2 (1-u)\right]_{\bar{u}}^1 + \int_{\bar{u}}^1
  G_\nu(u)^2\rd u 
\end{align*}
using integration by parts. 

Since $0  \leq G_\nu(u)^2 (1-u) \leq C^2
(1-u)^{2\epsilon}$ for $u \geq \bar{u}$ and $(1-u)^{2\epsilon} \to 0$ as $u \to 1$, it follows that 
$\left[ G_\nu(u)^2 (1-u)\right]_{\bar{u}}^1 = -G_\nu(\bar{u})^2(1-\bar{u})$.
Moreover, the second term is finite because
\begin{displaymath}
   \int_{\bar{u}}^1
  G_\nu(u)^2\rd u \leq  C^2 \int_{\bar{u}}^1 (1-u)^{-1+2\epsilon} \rd
  u = \frac{1}{2\epsilon}(1-\bar{u})^{2\epsilon}\,.
\end{displaymath}

\subsection{Proof of Theorem~\ref{theorem:spectr-transf-pit}}  
%\begin{proof}
Let $p_t$ denote the realized value of $P_t$ and $w_{t,j} =
G_j(p_t)$ the corresponding realized value of $W_{t,j}$ for
$t=1,\ldots,n$ and $j=1,2$. 
There are two cases to consider. Either $p_t$ occurs in an interval 
where the right derrivative of $G_j$ is 0 or in an interval
where the right derivative is positive. Let $\mathcal{G}_j$ denote the
subset of $[0,1]$ consisting of all points for which the right derivative  of $G_j$ equals zero.

If $p_t \in \mathcal{G}_j$ then, by the right-continuity of $G_j$, $p_t$ must occur in an interval of the form 
$[a_{t,j},b_{t,j})$ (if
there is a jump in $G_j$ at $b_{t,j}$) or $[a_{t,j},b_{t,j}]$ (if $G_j$ is continuous at
$b_{t,j}$). In either case the contribution of $w_{t,j}$ to the likelihood
is 
\begin{align*}
  \P\left(W_{t,j} = w_{t,j}\right) &= \P\left( G_j(P_t) = G_j(p_t) \right) 
                                     = F_P(b_{t,j}) - F_P(a_{t,j})\,.
\end{align*}
If $P_t \not\in \mathcal{G}_j$ then $w_{t,j}$ satisfies $\P(W_{t,j} \leq w_{t,j}) = \P(P_t\leq p_t)$ 
and $p_t = G_j^{-1}(w_{t,j})$, the unique inverse of $G_j$ at $w_{t,j}$. The contribution to the likelihood
is a density contribution given by
\begin{align*}
    f_W(w_{t,j}\mid\bm{\theta}) &=  \frac{f_P(p_t\mid\bm{\theta})}{G_j^\prime(p_t)}\,.
\end{align*}
The general form of the realized likelihood given $\bm{w}_j =
(w_{1,j},\ldots,w_{n,j})^\prime$ is thus
\begin{align*}
  \likhood_{W_j}(\bm{\theta} \mid \bm{w}_j) &= \prod_{p_t \in \mathcal{G}_j}(F_P(b_{t,j}) - F_P(a_{t,j}))
  \prod _{p_t \not\in \mathcal{G}_j} 
\frac{f_P(p_t)\mid\bm{\theta})}{G_j^\prime(p_t)}
\end{align*}
For the measures $\nu_1$ and $\nu_2$ the sets $\mathcal{G}_1$ and
$\mathcal{G}_2$ may differ at most by a null set.
Let us assume that each realized point $p_t$ is either in both of the sets
$\mathcal{G}_1$ and $\mathcal{G}_2$ or in neither of the sets.

If $p_t \in \mathcal{G}_1$ and $p_t \in \mathcal{G}_2$ then the agreement of the supports on
$(0,1)$ implies that $a_{t,1}=a_{t,2}$ and $b_{t,1}=b_{t,2}$. Thus
the likelihood contributions are identical.

If $p_t \not\in \mathcal{G}_1$ and $p_t \not\in \mathcal{G}_2$ then 
the likelihoods differ only by the scaling factor $G_j^\prime(p_t)$
which does not involve the parameters $\bm{\theta}$. This factor will appear in the
log-likelihood only as an unimportant additive term and cancel out of
the LR test statistic.

It follows that the likelihoods $\likhood_{W_j}(\bm{\theta} \mid \bm{w}_j)$
are maximized by the same values $\hat{\bm{\theta}}$ and the LR-test
statistics are identical.
%\end{proof}

\subsection{Proof of Theorem~\ref{theorem:pearson-test}}
% The Pearson test is one of the best known tests in statistics. The
% result can be proved by adapting an approach that is 
% used to derive the asymptotic distribution of the Pearson test
% statistic.

Let $\bm{X}_t = (X_{t,0},\ldots,X_{t,m})^\prime$ be the $(m+1)$-dimensional random vector with $X_{t,i} =\indicator{ \bm{1}^\prime \bm{W}_t =i}$ for
$i=0,\ldots,m$. Under~\eqref{eq:nullhypothesis} $\bm{X}_t$ has a
multinomial distribution satisfying $\E(X_{t,i})= \theta_i$,
$\var(X_{t,i}) = \theta_i(1-\theta_i)$ and $\cov(X_{t,i},X_{t,j})
=-\theta_i\theta_j$ for $i\ne j$.

Now define $\bm{Y}_t$ to be the $m$-dimensional random
vector obtained from $\bm{X}_t$ by omitting the first component. Then
$\E(\bm{Y}_t)= \bm{\theta}= (\theta_1,\ldots,\theta_m)^\prime$ and
$\Sigma_Y$ is the $m \times m$ submatrix of $\cov(\bm{X}_t)$ resulting
from deletion of the first row and column. Let $\average{\bm{Y}} = n^{-1} \sum_{t=1}^n \bm{Y}_t$.
A standard approach to the asymptotics of the Pearson test is to show
that
\begin{equation}\label{eq:8}
  S_m = \sum_{i=0}^m \frac{(O_i - n
    \theta_i)^2}{n\theta_i}= \sum_{i=0}^m \frac{(\sum_{t=1}^n X_{t,i} - n
    \theta_i)^2}{n\theta_i} = n (\average{\bm{Y}}- \bm{\theta})^\prime
  \Sigma_Y^{-1} (\average{\bm{Y}}- \bm{\theta}),
\end{equation}
and hence to
argue that $S_m \sim \chi^2_m$ in the limit as $n\to\infty$ by the central
limit theorem. It remains to show that the right-hand side of~\eqref{eq:8} has the
spectral test representation~(\ref{eq:28}).

Let $A$ be the $m \times m$ matrix with rows given by
$(\bm{e}_1-\bm{e}_2,\bm{e}_2-\bm{e}_3,\ldots,\bm{e}_m)$
where $\bm{e}_i$ denotes the $i$th unit vector. It may be easily
verified that $\bm{Y}_t = A\bm{W}_t$, $\bm{\theta}  = A\bm{\mu}_W$ and
$\Sigma_Y = A\Sigma_W A^\prime$.  It follows that
\begin{displaymath}
  n (\average{\bm{W}}- \bm{\mu}_W)^\prime
 \Sigma_W^{-1} (\average{\bm{W}}- \bm{\mu}_W) =   n (\average{\bm{Y}}- \bm{\theta})^\prime
  \Sigma_Y^{-1} (\average{\bm{Y}}- \bm{\theta}) = S_m.
\end{displaymath}

\section{Conditional bispectral  Z-test}
\label{app:cond-bisp-test}

The conditional spectral Z-test generalizes to a conditional
multispectral Z-test.  In the bispectral case, we construct 
two sets of transformed reported PIT-values
$(W_{t,1},W_{t,2})$ for $t=1,\ldots,n$, and form the vector
$\bm{Y}_t$ of length $k_1+k_2+2$ given by
\begin{equation}\label{eq:48}
  \bm{Y}_t =
  \left(\bm{h}^\prime_{t-1,1}\zeromean{W}_{t,1},\bm{h}^\prime_{t-1,2}\zeromean{W}_{t,2}\right)^\prime,
\end{equation}
where $\zeromean{W}_{t,i} = W_{t,i}-\mu_{W,i}$ and $\bm{h}_{t-1,i} =
(1,\lagpfunc_i(P_{t-1}),\ldots,\lagpfunc_i(P_{t-k_i}))^\prime$.   Parallel to the univariate case,
let $\average{\bm{Y}}_{n,k} = (n-k)^{-1}\sum_{t=k+1}^n \bm{Y}_t$ for $k=k_1\vee k_2$, 
and let $\hat{\Sigma}_Y$ denote a consistent estimator of
$\Sigma_{Y}:= \cov(\bm{Y}_t)$.  By the theory
of~\citet{bib:giacomini-white-06}, for $n$ large and $(k_1,k_2)$ fixed,
\begin{equation}\label{eq:32b}
 (n-k)\; \average{\bm{Y}}_{n,k}^\prime \; \hat{\Sigma}_{Y}^{-1} \; 
\average{\bm{Y}}_{n,k} \sim \chi^2_{k_1+k_2+2}.
\end{equation}

Working under the null hypothesis, we can generalize \eqref{eq:SigmaYnull} to 
$\Sigma_{Y} = A_W\circ \lagpfuncmat$, 
where $\circ$ denotes element-by-element multiplication (Hadamard product). The matrices are
\begin{equation}
\lagpfuncmat = \begin{pmatrix}
               \E\left( \bm{h}_{t-1,1}\bm{h}_{t-1,1}^\prime \right)
                                 & \E\left( \bm{h}_{t-1,1} \bm{h}_{t-1,2}^\prime \right) \\
               \E\left( \bm{h}_{t-1,2}\bm{h}_{t-1,1}^\prime \right) 
                                 &  \E\left( \bm{h}_{t-1,2} \bm{h}_{t-1,2}^\prime \right)
                       \end{pmatrix} ,\;\;
A_W = \begin{pmatrix}
            \sigma_{W,1}^2 \onesmatrix_{k_1+1,k_1+1} & \sigma_{W,12} \onesmatrix_{k_1+1,k_2+1} \\
            \sigma_{W,12} \onesmatrix_{k_2+1,k_1+1} & \sigma_{W,2}^2 \onesmatrix_{k_2+1,k_2+1}
           \end{pmatrix}  \label{eq:AsubW}
\end{equation}
where $\onesmatrix_{m,n}$ denotes the $m\times n$ matrix of ones and 
$\sigma_{W,12}=\E\left(\zeromean{W}_{t,1}\zeromean{W}_{t,2}\right)$.
Our tests use the estimator
$\hat{\Sigma}_{Y} = A_W\circ \hat{\lagpfuncmat}$, where $\hat{\lagpfuncmat}$ generalizes 
\eqref{eq:47} as 
\begin{equation}
\label{eq:47bispectral}
\hat{\lagpfuncmat} = (n-(k_1\vee k_2))^{-1} 
    \sum_{t=(k_1\vee k_2)+1}^n 
               (\bm{h}'_{t-1,1},\bm{h}'_{t-1,2})^\prime(\bm{h}'_{t-1,1},\bm{h}'_{t-1,2}).
\end{equation}

\section{Identification of spurious PIT values}
\label{app:spurious}

Consider a stylized Gaussian model in which loss is given by $L_t=\sigma_{t-1}Z_t$, 
where $(Z_t)$ is an iid sequence of standard normal random variables
and volatility $\sigma_{t-1}$ is $\bankfiltration_{t-1}$-measurable.  
Time variation in $\sigma_t$ may arise from stochastic volatility or from changes over time in portfolio composition. 
Suppose that the risk-manager knows the true underlying distribution and the volatility. 
The risk-manager's ideal value-at-risk forecast  at $\alpha=0.99$ is then $\widehat{\VaR}_{t} = \Phi^{-1}(0.99)\sigma_{t-1}$, 
where $\Phi$ is the standard normal cdf.   
We do not observe $\sigma_{t-1}$, but from observing
 $L_t$ and  $\widehat{\VaR}_{t}$, we can back out the realized value of $Z_t$ as
\begin{equation}
\label{eq:Zdefn}
 Z_t = \Phi^{-1}(0.99)\times L_t/\widehat{\VaR}_t.
\end{equation} 
Furthermore, the PIT values can be expressed as 
\begin{equation}
\label{eq:PfromZ}
P_t = \modelcdf_{t-1}(L_t) = \Phi( L_t / \sigma_{t-1}) = \Phi(Z_t).
\end{equation}

In general, we would not expect the $Z_t$ to be Gaussian, 
so \eqref{eq:PfromZ} will not hold.  However, so long as $(Z_t)$ is iid,
there will still be a monotonic relationship between $Z_t$ (as defined by \eqref{eq:Zdefn})
and $P_t$.   We find that the predicted relationship holds qualitatively for all bank-reported
portfolios, but with more noise in some portfolios than in others.  This suggests that we can use violations of monotonicity to identify spurious PIT values, but the threshold for identification must vary across portfolios.

Let $H(z; \theta_i): \mathbb{R}\rightarrow [0,1]$ be a family of fitting functions with parameter $\theta_i$ for portfolio $i$, and replace \eqref{eq:PfromZ} by
\begin{equation}
\label{eq:PfromHZ}
P_{i,t} = H(Z_{i,t}; \theta_i) + \epsilon_{i,t}
\end{equation}
where the $\epsilon_{i,t}$ are white-noise residuals.  Since the $H$ function should be increasing, it is convenient to take $H$ to be a cdf, even though it does not have a statistical interpretation in our context.  For convenience, we take $H$ to be the normal cdf with unrestricted $(\mu_i,\sigma_i)$ as $\theta_i$. 

For each portfolio $i$, we proceed as follows:
\begin{enumerate}
\item  Fit $\theta_i$ by nonlinear least squares, and construct residuals
$\epsilon_{it} = P_{it} - H(Z_{it}; \hat{\theta}_i)$.  
\item The $(\epsilon_{it})$ are bounded in the open interval $(-1,1)$, because $H(Z_{it})$ does not produce boundary values.  We model $\epsilon_{it}$ as drawn from a rescaled beta distribution on $(-1,1)$ with parameters 
$(\betaa=\tau_i/2, \betab=\tau_i/2)$.  This distribution has mean zero and variance $1/(\tau_i+1)$, so
we simply fit $\tau_i$ to the variance of the regression residuals.  
\item Let $B(\epsilon; \hat{\tau}_i)$ be the fitted beta distribution.  We flag an observation
$P_{it}$ as spurious whenever $B(\epsilon_{it}; \hat{\tau}_i) < q/2$ or $B(\epsilon_{it}; \hat{\tau}_i) > 1-q/2$, where $q$ is a tolerance parameter.
\item We reestimate $\tau_i$ as in step 3 on a sample that excludes the spurious observations.  Repeat step 4 with the updated $\hat{\tau}_i$.  An observation is flagged as spurious if it is rejected in \textit{either} round of estimation.
\end{enumerate}

In our baseline procedure, we set the tolerance parameter to  $q=10^{-5}$, which 
is intended to flag only the most egregious inconsistencies between $P_{it}$ and the pair 
$(L_{it}, \widehat{\VaR}_{it})$.  A typical case involves a PIT value very close to zero or one associated with a modest \PnL\ such that $\vert L_{it}\vert < \widehat{\VaR}_{it}$.  Setting $q=0$ is equivalent to shutting down the identification of spurious values.

The procedure yields \textit{imputed} PIT values as  $\imputed{P}_{it} = H(Z_{it}; \hat{\theta}_i)$.  As noted in Section 
\ref{sec:emp:conditional}, we use the imputed values to fill in for spurious values in forming regressors in the tests of conditional coverage.  

\end{doublespacing}

\bibliographystyle{jf}  
\bibliography{spectraltest}

\newcommand{\noopsort}[1]{}
\begin{thebibliography}{39}
\expandafter\ifx\csname natexlab\endcsname\relax\def\natexlab#1{#1}\fi

\bibitem[Acerbi and Szekely(2014)]{bib:acerbi-szekely-14}
Acerbi, C., and B.~Szekely, 2014, Back-testing expected shortfall, {\em Risk\/}
   1--6.

\bibitem[Amisano and Giacomini(2007)]{bib:amisano-giacomini-07}
Amisano, G., and R.~Giacomini, 2007, Comparing density forecasts via weighted
  likelihood ratio tests, {\em {Journal of Business \& Economic Statistics}\/}
  25, 177--190.

\bibitem[Andrews(1991)]{bib:andrews-91}
Andrews, D.W.K., 1991, Heteroskedasticity and autocorrelation consistent
  covariance matrix estimation, {\em Econometrica\/} 59, 817--858.

\bibitem[Barone-Adesi et~al.(1998)Barone-Adesi, Bourgoin, and
  Giannopoulos]{bib:barone-adesi-et-al-98}
Barone-Adesi, G., F.~Bourgoin, and K.~Giannopoulos, 1998, Don't look back, {\em
  Risk\/} 11, 100--103.

\bibitem[{Basel Committee on Bank Supervision}(2013)]{bib:basel-13}
{Basel Committee on Bank Supervision}, 2013, Fundamental review of the trading
  book: {A} revised market risk framework, Publication No. 265, {Bank for
  International Settlements}.

\bibitem[Berkowitz(2001)]{bib:berkowitz-01}
Berkowitz, J., 2001, Testing the accuracy of density forecasts, applications to
  risk management, {\em {Journal of Business \& Economic Statistics}\/} 19,
  465--474.

\bibitem[Berkowitz and O'Brien(2002)]{bib:berkowitz-o-brien-02}
Berkowitz, J., and J.~O'Brien, 2002, How accurate are {Value-at-Risk} models at
  commercial banks?, {\em {The Journal of Finance}\/} 57, 1093--1112.

\bibitem[Billingsley(1961)]{bib:billingsley-61}
Billingsley, P., 1961, The {Lindeberg--L\'{e}vy} theorem for martingales, {\em
  Proceedings of the American Mathematical Society\/} 12, 788--792.

\bibitem[{Board of Governors of the Federal Reserve
  System}(2011)]{bib:us-sr-11-7}
{Board of Governors of the Federal Reserve System}, 2011, Supervisory guidance
  on model risk management, SR Letter 11-7.

\bibitem[Campbell(2006)]{bib:campbell-06}
Campbell, S.D., 2006, A review of backtesting and backtesting procedures, {\em
  {Journal of Risk}\/} 9, 1--17.

\bibitem[Christoffersen(1998)]{bib:christoffersen-98}
Christoffersen, P., 1998, Evaluating interval forecasts, {\em {International
  Economic Review}\/} 39.

\bibitem[Colletaz et~al.(2013)Colletaz, Hurlin, and P\'erignon]{Colletaz2013}
Colletaz, Gilbert, Christophe Hurlin, and Christophe P\'erignon, 2013, The risk
  map: A new tool for validating risk models, {\em {Journal of Banking and
  Finance}\/} 37, 3843--3854.

\bibitem[Costanzino and Curran(2015)]{bib:costanzino-curran-15}
Costanzino, N., and M.~Curran, 2015, Backtesting general spectral risk measures
  with application to expected shortfall, {\em The Journal of Risk Model
  Validation\/} 9, 21--31.

\bibitem[Crnkovic and Drachman(1996)]{bib:crnkovicDrachman-96}
Crnkovic, C., and J.~Drachman, 1996, Quality control, {\em Risk\/} 9, 139--143.

\bibitem[Diebold et~al.(1998)Diebold, Gunther, and
  Tay]{bib:diebold-gunther-tay-98}
Diebold, F.X., T.A. Gunther, and A.S. Tay, 1998, Evaluating density forecasts
  with applications to financial risk management, {\em {International Economic
  Review}\/} 39, 863--883.

\bibitem[Diebold and Mariano(1995)]{bib:diebold-mariano-95}
Diebold, F.X., and R.S. Mariano, 1995, Comparing predictive accuracy, {\em
  {Journal of Business \& Economic Statistics}\/} 13, 253--265.

\bibitem[Du and Escanciano(2017)]{bib:du-escanciano-17}
Du, Z., and J.C. Escanciano, 2017, Backtesting expected shortfall: accounting
  for tail risk, {\em {Management Science}\/} 63, 940--958.

\bibitem[Engle and Manganelli(2004)]{bib:engle-manganelli-04}
Engle, R.F., and S.~Manganelli, 2004, {CAViaR}: conditional autoregressive
  value at risk by regression quantiles, {\em {Journal of Business \& Economic
  Statistics}\/} 22, 367--381.

\bibitem[Escanciano and Olmo(2010)]{bib:escanciano-olmo-10}
Escanciano, J.C., and J.~Olmo, 2010, Backtesting parametric value-at-risk with
  estimation risk, {\em Journal of Business \& Economic Statistics\/} 28,
  36--51.

\bibitem[Federal Register(2012)]{bib:rbcg:marketrisk}
Federal Register, 2012, Risk-based capital guidelines: Market risk.

\bibitem[Fissler et~al.(2016)Fissler, Ziegel, and
  Gneiting]{bib:fissler-ziegel-gneiting-16}
Fissler, T., J.F. Ziegel, and T.~Gneiting, 2016, Expected shortfall is jointly
  elicitable with value-at-risk: implications for backtesting, {\em Risk\/}
  58--61.

\bibitem[Giacomini and White(2006)]{bib:giacomini-white-06}
Giacomini, R., and H.~White, 2006, Tests of conditional predictive ability,
  {\em Econometrica\/} 74, 1545--1578.

\bibitem[Gneiting(2011)]{bib:gneiting-11}
Gneiting, T., 2011, Making and evaluating point forecasts, {\em {Journal of the
  American Statistical Association}\/} 106, 746--762.

\bibitem[Gneiting et~al.(2007)Gneiting, Balabdaoui, and
  Raftery]{bib:gneiting-balabdaoui-raftery-07}
Gneiting, T., F.~Balabdaoui, and A.E. Raftery, 2007, Probabilistic forecasts,
  calibration and sharpness, {\em Journal of the Royal Statistical Society,
  Series B\/} 69, 243--268.

\bibitem[Gneiting and Ranjan(2011)]{bib:gneiting-ranjan-11}
Gneiting, T., and R.~Ranjan, 2011, Comparing density forecasts using threshold-
  and quantile-weighted scoring rules, {\em {Journal of Business \& Economic
  Statistics}\/} 29, 411--422.

\bibitem[Hewitt(1960)]{bib:hewitt-60}
Hewitt, E., 1960, Integration by parts for {S}tieltjes integrals, {\em The
  American Mathematical Monthly\/} 67, 419--423.

\bibitem[Hull and White(1998)]{bib:hull-white-98}
Hull, J.~C., and A.~White, 1998, Incorporating volatility updating into the
  historical simulation method for {Value-at-Risk}, {\em {Journal of Risk}\/}
  1, 5--19.

\bibitem[Hurlin et~al.(2017)Hurlin, Laurent, Quaedvlieg, and
  Smeekes]{bib:hurlin-et-al-17}
Hurlin, C., S.~Laurent, R.~Quaedvlieg, and S.~Smeekes, 2017, Risk measure
  inference, {\em Journal of Business \& Economic Statistics\/} 35, 499--512.

\bibitem[Kratz et~al.(2018)Kratz, Lok, and McNeil]{bib:kratz-lok-mcneil-18}
Kratz, M., Y.H. Lok, and A.J. McNeil, 2018, {Multinomial VaR backtests}: A
  simple implicit approach to backtesting expected shortfall, {\em Journal of
  Banking and Finance\/} 88, 393--407.

\bibitem[Kupiec(1995)]{bib:kupiec-95}
Kupiec, P.~H., 1995, Techniques for verifying the accuracy of risk measurement
  models, {\em {Journal of Derivatives}\/} 3, 73--84.

\bibitem[Leccadito et~al.(2014)Leccadito, Boffelli, and Urga]{Leccadito2014}
Leccadito, Arturo, Simona Boffelli, and Giovanni Urga, 2014, Evaluating the
  accuracy of {V}alue-at-{R}isk forecasts: New multilevel tests, {\em
  {International Journal of Forecasting}\/} 30, 206--216.

\bibitem[Newey and West(1987)]{bib:newey-west-87}
Newey, W., and K.~West, 1987, A simple, positive semi-definite,
  heteroskedasticity and autocorrelation consistent covariance matrix, {\em
  Econometrica\/} 55, 703--08.

\bibitem[Nolde and Ziegel(2017)]{bib:nolde-ziegel-17}
Nolde, N., and J.F. Ziegel, 2017, Elicitability and backtesting: Perspectives
  for banking regulation, {\em Annals of Applied Statistics\/} 11, 1833--1874.

\bibitem[O'Brien and Szerszen(2017)]{bib:obrien-szerszen-17}
O'Brien, J., and P.J. Szerszen, 2017, An evaluation of bank measures for market
  risk before, during and after the financial crisis, {\em {Journal of Banking
  and Finance}\/} 80, 215--234.

\bibitem[P\'erignon et~al.(2008)P\'erignon, Deng, and
  Wang]{bib:perignon-deng-wang-08}
P\'erignon, C., Z.Y. Deng, and Z.J. Wang, 2008, Diversification and
  {Value-at-Risk}, {\em {Journal of Banking and Finance}\/} 32, 783--794.

\bibitem[P\'erignon and Smith(2010)]{bib:perignon-smith-10}
P\'erignon, C., and D.~R. Smith, 2010, The level and quality of {Value-at-Risk}
  disclosure by commercial banks, {\em {Journal of Banking and Finance}\/} 34,
  362--377.

\bibitem[P\'{e}rignon and Smith(2008)]{bib:perignon-smith-08}
P\'{e}rignon, C., and D.R. Smith, 2008, A new approach to comparing {VaR}
  estimation methods, {\em Journal of Derivatives\/} 16, 54--66.

\bibitem[Rosenblatt(1952)]{bib:rosenblatt-52}
Rosenblatt, M., 1952, Remarks on a multivariate transformation, {\em Annals of
  Mathematical Statistics\/} 23, 470--472.

\bibitem[Shaffer(1995)]{bib:schaffer-95}
Shaffer, J.~P., 1995, Multiple hypothesis testing, {\em Annual Review of
  Psychology\/} 46, 561--584.

\end{thebibliography}


\newcommand{\noopsort}[1]{}
\begin{thebibliography}{2}
\expandafter\ifx\csname natexlab\endcsname\relax\def\natexlab#1{#1}\fi

\bibitem[Abramowitz and Stegun(1965)]{bib:abramowitz-stegun-65}
Abramowitz, M., and I.~A. Stegun, eds., 1965, {\em Handbook of Mathematical
  Functions\/} (Dover Publications, New York).

\bibitem[Milgram(2010)]{hypergeom3F2:review}
Milgram, Michael~S., 2010, On hypergeometric 3{F}2(1) - a review, Working Paper
  1011.4546, arXiv.

\end{thebibliography}

\end{document}

\begin{titlepage} \singlespacing

\maketitle
\begin{abstract}
We extract material from a companion paper (in progress) which elaborates on certain results in our main paper.  To avoid confusion with references to tables, figures and equations in the main paper, we prepend ``S'' when numbering tables, figures and equations contained in this supplement.
\end{abstract}
\bigskip
\thispagestyle{empty}
\end{titlepage}
\pagestyle{plain} \setcounter{page}{1}

%% Change the numbering scheme.
\renewcommand{\theequation}{S.\arabic{equation}}
\renewcommand{\thetable}{S.\arabic{table}}
\renewcommand{\thefigure}{S.\arabic{figure}}
\renewcommand{\thetheorem}{S.\arabic{theorem}}
\renewcommand{\theproposition}{S.\arabic{proposition}}

\begin{onehalfspacing}
\renewcommand{\thesection}{Supplement \Alph{section}:}
\section{Truncated probitnormal score test}\label{app:pns}
 The probitnormal model yields a further new spectral
 test based on a classical score test of $\bm{\theta}=\bm{\theta}_0$ against the alternative
 $\bm{\theta}\neq\bm{\theta}_0$.
 Let $\likhood_P(\bm{\theta}\mid P_t^*)$ denote the likelihood
 contribution of a truncated observation $P_t^*=\alpha_1 \vee (P_t
 \wedge \alpha_2)$ when $P_t$ follows~(13) and write
 \begin{equation}\label{eq:38}
   \bm{S}_t(\bm{\theta}) = \left(\frac{\partial}{\partial \mu} 
     \ln \likhood_P(\bm{\theta}\mid P_t^*), \frac{\partial}{\partial \sigma} 
      \ln \likhood_P(\bm{\theta}\mid P_t^*)\right)^\prime
 \end{equation}
 for the corresponding score vector. Let $\average{\bm{S}}_n(\bm{\theta}_0) =\tfrac{1}{n}\sum_{t=1}^n
 \bm{S}_t(\bm{\theta}_0)$ be the mean of the observed score vectors
 under the null. 

 Standard likelihood theory implies that $\sqrt{n}\average{\bm{S}}_n(\bm{\theta}_0)
     \xrightarrow[n\to \infty]{d} N_2\big(\bm{0},I(\bm{\theta}_0)\big)$
     under the null,
     where $I(\bm{\theta})$ denotes the covariance matrix of
     $\bm{S}_t(\bm{\theta})$, i.e.,~the Fisher information matrix.
 For large $n$ we have approximately that
 \begin{equation}
 \label{eq:scoretest}
   n \average{\bm{S}}_n(\bm{\theta}_0)^\prime I(\bm{\theta}_0)^{-1}
   \average{\bm{S}}_n(\bm{\theta}_0) \sim \chi^2_2 \,.
 \end{equation}
An analytical expression for $I(\bm{\theta}_0)$ is provided later in this appendix.
 
 Under a restriction on the kernel window we can show that the score
 test~\eqref{eq:scoretest} is a bispectral Z-test with kernel
 measures $\nu_1$ and $\nu_2$ given by sums of discrete and continuous parts.
 \begin{theorem}\label{prop:prob-score-test}
 Let $z_0$ be the unique solution to the equation
 \begin{equation}
   \label{eq:6z}
   z^2 + \left(\phi(z)/\Phi(z)\right)z - 1 = 0.
 \end{equation}
 Provided $\Phi(z_0)\leq\alpha_1<\alpha_2\leq 1$, then 
 $\bm{S}_t(\bm{\theta}_0) = \bm{W}_t -\bm{\mu}_W$, almost surely, where
 \begin{equation}\label{eq:9}
   W_{t,i} = \dnudiscrete_{i,1}\indicator{P_t \geq \alpha_1} + \dnudiscrete_{i,2}\indicator{P_t \geq
   \alpha_2} + \int_{\alpha_1}^{\alpha_2}g_i(u)\indicator{P_t\geq u}\rd u
 \end{equation}
 for $\dnudiscrete_{i,1}\geq 0$, $\dnudiscrete_{i,2}\geq 0$ and
 $g_i(u)$ positive and differentiable on $[\alpha_1, \alpha_2]$.
 \end{theorem}
\begin{proof}
  % The likelihood $L_P(\bm{\theta} \mid \bm{P}^*)$ takes the form 
 \begin{equation}\label{eq:7}
   L_P(\bm{\theta} \mid \bm{P}^*) = \prod_{t\,:\,P_t^* = \alpha_1}
   F_P(\alpha_1\mid\bm{\theta})
   \prod _{t\,:\,\alpha_1 < P_t^* < \alpha_2} f_P(P_t^*\mid\bm{\theta})
   \prod_{t\,:\,P_t^* = \alpha_2} \tail{F}_P(\alpha_2\mid\bm{\theta})
 \end{equation}
 where $\tail{F}(u)$ denotes the tail probability $1-F(u)$.
 The likelihood contributions 
 $\likhood_P(\bm{\theta}\mid P_t^*)$ are given by the individual terms
 in~\eqref{eq:7} according to whether $P_t^*=\alpha_1$, $\alpha_1 <
 P_t^* < \alpha_2$ or $P_t^*=\alpha_2$. 
 Computing the score statistic and evaluating it at $\bm{\theta}_0 =
 (0,1)^\prime$ yields 
   \begin{displaymath}
     \bm {S}_t(\bm{\theta}_0) =
     \begin{cases}
      \bm{\scorecase}_1(\alpha_1) & P_t^* = \alpha_1,\\
      \bm{\scorecase}_*(P_t^*) & \alpha_1    < P_t^*  < \alpha_2, \\
      \bm{\scorecase}_2(\alpha_2)  & P_t^* = \alpha_2.
     \end{cases}
 \quad\text{where}\quad
  \bm{\scorecase}_1(u) = \left(\begin{array}{c} -\phi(\Phi^{-1}(u))/u \\
  -\phi(\Phi^{-1}(u))\Phi^{-1}(u)/u \end{array}\right),
   \end{displaymath}
   \begin{align*}
  \bm{\scorecase}_*(u) &= \left( \begin{array}{c} \Phi^{-1}(u) \\
       \Phi^{-1}(u)^2-1 \end{array} \right) \quad\text{and}\quad
  \bm{\scorecase}_2(u) = \left(\begin{array}{c} \phi(\Phi^{-1}(u))/(1-u)  \\
            \phi(\Phi^{-1}(u)) \Phi^{-1}(u)/(1-u) \end{array}\right) \,.
   \end{align*}
 The discontinuities at $\alpha_1$ and $\alpha_2$ are given by
  \begin{equation*}
  (\dnudiscrete_{1,1}, \dnudiscrete_{2,1})^\prime = \bm{\scorecase}_*(\alpha_1) - 
 \bm{\scorecase}_1(\alpha_1),
  \qquad
  (\dnudiscrete_{1,2}, \dnudiscrete_{2,2})^\prime = \bm{\scorecase}_2(\alpha_2) - 
 \bm{\scorecase}_*(\alpha_2)
   \end{equation*}
 and non-negativity of the $\dnudiscrete_{i,j}$ in all cases is guaranteed
 provided $\bm{\scorecase}_*(\alpha_1) - 
 \bm{\scorecase}_1(\alpha_1)\geq \bm{0}$. The second component of this
 vector inequality leads to condition~\eqref{eq:6z}.
 The weighting functions can be obtained by
   differentiating $\bm{\scorecase}_*(u)$ with respect to $u$ on
 $[\alpha_1,\alpha_2]$ and are thus
 \begin{displaymath}
   g_1(u) = \frac{1}{\phi(\Phi^{-1}(u))}\indicator{\alpha_1 \leq u \leq
     \alpha_2},\qquad g_2(u) = 
 \frac{2\Phi^{-1}(u)}{\phi(\Phi^{-1}(u))}\indicator{\alpha_1 \leq u \leq
     \alpha_2}.
 \end{displaymath}
 Finally, since $\bm{\mu}_{\bm{W}}  = \bm{W}_t -
 \bm{S}_t(\bm{\theta}_0)$, we must have that
  $\bm{\mu}_{\bm{W}}  = -\bm{\scorecase}_1(\alpha_1)$.
\end{proof}

 We find $\Phi(z_0)\approx 0.8$, so the constraint on $\alpha_1$ is unlikely to bind in application to the range of
 tail probability levels of practical interest.  For $\alpha_2 < 1$ the $W_{t,i}$ variables are bounded, guaranteeing
 that the elements of $I(\bm{\theta}_0)$ are finite. For $\alpha_2=1$
 the $G_\nu$ functions for $\nu_1$ and $\nu_2$ grow like $\Phi^{-1}(u)$
 and $\Phi^{-1}(u)^2$ respectively. We can use
 the asymptotic approximation $\Phi^{-1}(u) \sim \sqrt{-2\ln(1-u)}$ as
 $u \to 1$ to verify
 that the condition of Proposition 3.2
 is satisfied in both cases.

\subsection*{Fisher information matrix}
The following identities are useful for dealing with the probitnormal 
distribution:
\begin{align}
  \int_{\alpha_1}^{\alpha_2} \Phi^{-1}(u) \rd u &=
  \phi(\Phi^{-1}(\alpha_1)) -   \phi(\Phi^{-1}(\alpha_2)) \label{eq:40}\\
\int_{\alpha_1}^{\alpha_2} \left(\Phi^{-1}(u)^2 - 1\right) \rd u&=  
\Phi^{-1}(\alpha_1)\phi(\Phi^{-1}(\alpha_1)) -   
\Phi^{-1}(\alpha_2)\phi(\Phi^{-1}(\alpha_2)).\label{eq:41}
\end{align}
Let $\xi(p \mid \bm{\theta}) = (\Phi^{-1}(p)-\mu)/\sigma$. The first
derivatives of the log-likelihood of the truncated probitnormal
distribution are
\begin{equation}\label{eq:190}
\frac{\partial}{\partial
      \mu }\ln \likhood(\bm{\theta}\mid P_t^*)=
    \begin{cases}
-\frac{
\phi\big(\xi(\alpha_1 \mid \bm{\theta})\big)}{ \sigma\Phi\big(\xi(\alpha_1 \mid 
\bm{\theta})\big)}
 & P_t^* = \alpha_1,\\
-\frac{\xi\big(P_t^*\mid\bm{\theta}\big)}{\sigma}&
 \alpha_1    < P_t^*  < \alpha_2, \\
 \frac{
\phi\big(\xi(\alpha_2\mid \bm{\theta})\big)}{ 
\sigma\overline{\Phi}\big(\xi(\alpha_2 \mid \bm{\theta})\big)}
& P_t^* = \alpha_2,
    \end{cases}
  \end{equation}
and
\begin{equation}\label{eq:191}
\frac{\partial}{\partial
      \sigma }\ln \likhood(\bm{\theta}\mid P_t^*)=
    \begin{cases}
-\frac{
\phi\big(\xi(\alpha_1 \mid \bm{\theta})\big) \xi(\alpha_1 \mid \bm{\theta})}{ 
\sigma\Phi\big(\xi(\alpha_1 \mid \bm{\theta})\big)}
 & P_t^* = \alpha_1,\\
-\frac{\xi\big(P_t^*\mid\bm{\theta}\big)^2+1}{\sigma}&
 \alpha_1    < P_t^*  < \alpha_2, \\
 \frac{
\phi\big(\xi(\alpha_2\mid \bm{\theta})\big) \xi(\alpha_2 \mid \bm{\theta})}{ 
\sigma\overline{\Phi}\big(\xi(\alpha_2 \mid \bm{\theta})\big)}
& P_t^* = \alpha_2.
    \end{cases}
  \end{equation}

Recall that the expected Fisher information matrix is defined as  
\begin{displaymath}
  I(\bm{\theta})_{ij} = -\E \left( \frac{\partial^2}{\partial
      \theta_i \partial \theta_j}\ln \likhood(\bm{\theta}\mid P_t^*)\right).
\end{displaymath}
The conditional second derivatives of the log-likelihood are
 \begin{equation}\label{eq:390}
-\frac{\partial^2}{\partial
      \mu^2 }\ln \likhood(\bm{\theta}\mid P_t^*)=
    \begin{cases}
\frac{
\phi(\xi(\alpha_1 \mid \bm{\theta}))\Big( \phi(\xi(\alpha_1 \mid \bm{\theta})) 
+\xi(\alpha_1 \mid \bm{\theta})\Phi(\xi(\alpha_1
  \mid \bm{\theta})) \Big)}{ \sigma^2\Phi(\xi(\alpha_1 \mid \bm{\theta}))^2}
 & P_t^* = \alpha_1,\\
\frac{1}{\sigma^2}&
 \alpha_1    < P_t^*  < \alpha_2, \\
 \frac{
\phi(\xi(\alpha_2\mid \bm{\theta}))\Big(\phi(\xi(\alpha_2 \mid \bm{\theta})) - 
\xi(\alpha_2 \mid \bm{\theta})\overline{\Phi}(\xi(\alpha_2
  \mid \bm{\theta}))   \Big)}{ \sigma^2\overline{\Phi}(\xi(\alpha_2 \mid 
\bm{\theta}))^2}
& P_t^* = \alpha_2,
    \end{cases}
  \end{equation}
\begin{equation}\label{eq:391}
-\frac{\partial^2}{\partial
      \sigma^2 }\ln \likhood(\bm{\theta}\mid P_t^*)=
    \begin{cases}
\frac{
\phi(\xi(\alpha_1 \mid \bm{\theta}))\Big(  \xi(\alpha_1 \mid 
\bm{\theta})^2\phi(\xi(\alpha_1 \mid \bm{\theta})) +\xi(\alpha_1 \mid 
\bm{\theta})^3\Phi(\xi(\alpha_1 \mid \bm{\theta})) -2\xi(\alpha_1 \mid 
\bm{\theta})\Phi(\xi(\alpha_1
  \mid \bm{\theta})) \Big)}{ \sigma^2\Phi(\xi(\alpha_1 \mid \bm{\theta}))^2}
 & P_t^* = \alpha_1,\\
\frac{3\xi(P_t^*\mid \bm{\theta})^2-1}{\sigma^2}&
 \alpha_1    < P_t^*  < \alpha_2, \\
 \frac{
\phi(\xi(\alpha_2\mid \bm{\theta}))
\Big( \xi(\alpha_2 \mid
  \bm{\theta})^2\phi(\xi(\alpha_2 \mid \bm{\theta})) - \xi(\alpha_2
  \mid \bm{\theta})^3\overline{\Phi}(\xi(\alpha_2 \mid \bm{\theta})) +
  2\xi(\alpha_2 \mid \bm{\theta})\overline{\Phi}(\xi(\alpha_2
  \mid \bm{\theta}))\Big)
}{ \sigma^2\overline{\Phi}(\xi(\alpha_2 \mid \bm{\theta}))^2}
& P_t^* = \alpha_2,
    \end{cases}
  \end{equation}
\begin{equation}\label{eq:392}
-\frac{\partial^2}{\partial
      \mu \partial \sigma }\ln \likhood(\bm{\theta}\mid P_t^*)=
    \begin{cases}
\frac{
\phi(\xi(\alpha_1 \mid \bm{\theta}))\Big( \phi(\xi(\alpha_1 \mid \bm{\theta}))
  \xi(\alpha_1 \mid \bm{\theta}) -\Phi(\xi(\alpha_1 \mid \bm{\theta}))
  + \xi(\alpha_1 \mid \bm{\theta})^2\Phi(\xi(\alpha_1
  \mid \bm{\theta})) \Big)}
{ \sigma^2\Phi(\xi(\alpha_1 \mid \bm{\theta}))^2}
 & P_t^* = \alpha_1,\\
\frac{2\xi(P_t^*\mid \bm{\theta})}{\sigma^2}&
 \alpha_1    < P_t^*  < \alpha_2, \\
 \frac{
\phi(\xi(\alpha_2\mid \bm{\theta}))\Big( \phi(\xi(\alpha_2 \mid \bm{\theta}))
  \xi(\alpha_2 \mid \bm{\theta}) + \overline{\Phi}(\xi(\alpha_2 \mid
  \bm{\theta})) - \xi(\alpha_2 \mid \bm{\theta})^2\overline{\Phi}(\xi(\alpha_2
  \mid \bm{\theta}))\Big)}{ \sigma^2\overline{\Phi}(\xi(\alpha_2 \mid 
\bm{\theta}))^2}
& P_t^* = \alpha_2.
    \end{cases}
  \end{equation}
By taking expectations using~\eqref{eq:40} and~\eqref{eq:41} and evaluating at 
$\bm{\theta}_0 = (0,1)^\prime$ 
we obtain the elements of $I(\bm{\theta}_0)$:
  \begin{multline}\label{eq:43}
    I(\bm{\theta}_0)_{1,1} = \phi(\Phi^{-1}(\alpha_1))^2/\alpha_1 +   
\phi(\Phi^{-1}(\alpha_2))^2/(1-\alpha_2)\\
+  \phi(\Phi^{-1}(\alpha_1))\Phi^{-1}(\alpha_1) 
-\phi(\Phi^{-1}(\alpha_2))\Phi^{-1}(\alpha_2) +
    (\alpha_2-\alpha_1),
  \end{multline}
 \begin{multline}\label{eq:44}
    I(\bm{\theta}_0)_{2,2} =
    \phi(\Phi^{-1}(\alpha_1))^2\Phi^{-1}(\alpha_1)^2/\alpha_1 +
    \phi(\Phi^{-1}(\alpha_1))\Phi^{-1}(\alpha_1)^3 \\
+ \phi(\Phi^{-1}(\alpha_1))\Phi^{-1}(\alpha_1) 
 + \phi(\Phi^{-1}(\alpha_2))^2\Phi^{-1}(\alpha_2)^2/(1-\alpha_2) \\ -
    \phi(\Phi^{-1}(\alpha_2))\Phi^{-1}(\alpha_2)^3
- \phi(\Phi^{-1}(\alpha_2))\Phi^{-1}(\alpha_2) 
+2(\alpha_2-\alpha_1),
\end{multline}
\begin{multline}\label{eq:45}
    I(\bm{\theta}_0)_{1,2} = \phi(\Phi^{-1}(\alpha_1))^2
    \Phi^{-1}(\alpha_1)/\alpha_1 + \phi(\Phi^{-1}(\alpha_1))\big(1+
    \Phi^{-1}(\alpha_1)^2\big) \\
+ \phi(\Phi^{-1}(\alpha_2))^2
    \Phi^{-1}(\alpha_2)/(1-\alpha_2) - \phi(\Phi^{-1}(\alpha_2))\big(1+
    \Phi^{-1}(\alpha_2)^2\big).
\end{multline}

\section{Moments for the beta kernel}
\label{app:beta}
We provide a general solution to the moments and cross-moments of the transformed PIT values
when the kernel densities take the form
\[
g_\nu(u) = \frac{(u-\alpha_1)^{\betaa-1} (\alpha_2-u)^{\betab-1}}{(\alpha_2-\alpha_1)^{\betaa+\betab-1} B(\betaa, \betab)} 
\]
for parameters $(\betaa>0, \betab>0)$ and $\alpha_1\leq u\leq\alpha_2$. 
The normalization guarantees that $G_\nu(\alpha_2)=1$, and helps align the solution with standard beta distribution functions provided by statistical packages.  In \texttt{R} notation, the kernel function is simply
\[
  G_\nu(u) = \operatorname{pbeta}\left(\frac{\max\{\alpha_1,\min\{u,\alpha_2\}\}-\alpha_1}
                                                                               {\alpha_2-\alpha_1}, \betaa, \betab\right).
\]
% \texttt{pbeta}$((u^*-\alpha_1)/(\alpha_2-\alpha_1),\betaa, \betab)$ where 
%$u^*=\alpha_1 \vee (u\wedge \alpha_2) $.

Solving for moments and cross-moments of kernels $(g_1(P), g_2(P))$ for uniform $P$ involves the following integral:
\begin{multline}
\label{eq:defineM} 
M(\betaa_1,\betab_1,\betaa_2,\betab_2) = 
\int_{\alpha_1}^{\alpha_2} (1-u) g_1(u) G_2(u) du \\
     =   \frac{B(\betaa_1+\betaa_2,1+\betab_1)}{\betaa_2 B(\betaa_1,\betab_1) B(\betaa_2,\betab_2)}
             {_3}F_2(\betaa_2,\betaa_1+\betaa_2,1-\betab_2; 1+\betaa_2, 1+\betaa_1+\betaa_2+\betab_1; 1)\\
      = \frac{B(\betaa_1+\betaa_2,1+\betab_1+\betab_2)}{\betaa_2 B(\betaa_1,\betab_1) B(\betaa_2,\betab_2)}
             {_3}F_2(1,\betaa_1+\betaa_2,\betaa_2+\betab_2; 1+\betaa_2, 1+\betaa_1+\betaa_2+\betab_1+\betab_2; 1)
\end{multline}
where ${_3}F_2(c_1,c_2,c_3; d_1, d_2; 1)$ denotes a hypergeometric function of order $(3,2)$ and argument unity.  The final line follows from the Thomae transformation T7 in 
\citet[Appendix A]{hypergeom3F2:review}. 
Due to the normalization of the kernels, $M$ does not depend on the choice of kernel window.

When its parameters are all positive, as in the final form in \eqref{eq:defineM}, numerical solution to  
${_3}F_2(c_1,c_2,c_3; d_1, d_2; 1)$ is straightforward via the standard hypergeometric series expansion.  In practice, we are most often interested in integer-valued cases for which $M$ has a simple closed-form solution.

For given kernel window and PIT value, let $W_{\betaa,\betab}$ be the transformed PIT value under a beta kernel with parameters $(\betaa,\betab)$.  A recurrence rule for the incomplete beta function
\citep[][eq.~6.6.7]{bib:abramowitz-stegun-65} leads to a linear relationship among 
``neighboring'' transformations:
\begin{equation}
\label{eq:betarecur}
(\betaa+\betab) W_{\betaa,\betab} = \betaa W_{\betaa+1,\betab} +  \betab W_{\betaa,\betab+1}
\end{equation}
An immediate implication is that the uniform, linear increasing and  linear decreasing transformations (parameter sets (1,1), (2,1) and (1,2), respectively) are linearly dependent.  
Any pair of these kernels would yield an equivalent bispectral test, and a trispectral test using all three kernels would be undefined due to a singular covariance matrix $\Sigma_W$.  By iterating the recurrence relationship, we can derive linear relationships among sets of kernels with integer-valued parameter differences $\betaa_i-\betaa_\ell$ and $\betab_i-\betab_\ell$, which would lead to redundancies among the corresponding $j$-spectral tests.

\section{Monte Carlo simulations}
\label{app:MCsimstudy}
To compare unconditional tests we generate pseudo PIT
values by first sampling from
standard normal and scaled Student $t_5$ and $t_3$ distributions, which represent
the true model. We then transform
the values to the interval $(0,1)$ using the standard normal
cdf which is taken to be the risk manager's model. The Student distributions are scaled to have
variance one so differences stem from different tail
shapes rather than different variances. The PIT samples arising from
normal are uniformly distributed and are used to evaluate
the size of the tests. The PIT samples arising from the Student $t$
distributions show the kind of departures from uniformity that are
observed when tails are poorly estimated. 

The majority of the kernel and test abbreviations are explained in Section~4.3
of the main paper. Additional mnemonics used in this Supplement are
\begin{description}
\item[5-level multinomial tests:]  we apply the Pearson test (PE5) and the 
 Z-test with discrete uniform kernel (ZU5) in addition to the 3-level
 tests used in the main paper.
\item[Discrete LR-tests:] the binomial
  LR-test (LR1) and the 3-level multinomial LR-test (LR3).
\item[Continuous LR-test:] the LR-test described in Section~4.2 of the main paper which
  may be viewed as an extension of the test of Berkowitz (LRB).
\end{description}
\noindent The baseline sample size is $n =750$, which corresponds approximately to the three-year 
samples of bank data in the main paper.  In all simulation experiments number of replications is $2^{16}=\num{65536}$.

\begin{table}[htb!]
  \centering
  \begin{tabular}{*{2}{l}*{5}{r}}
    \toprule
    window & \( F \) \textbar\ kernel & \multicolumn{1}{c}{BIN} & \multicolumn{1}{c}{ZU3} & \multicolumn{1}{c}{ZU5} & \multicolumn{1}{c}{PE3} & \multicolumn{1}{c}{PE5} \\
    \midrule
    narrow & Normal & 6.1 & 4.9 & 4.6 & 5.3 & 5.9 \\
    & Scaled t5 & 33.9 & 35.0 & 34.0 & 40.3 & 33.5 \\
    & Scaled t3 & 24.0 & 24.8 & 24.1 & 43.4 & 33.5 \\ \addlinespace[3pt]
    wide & Normal & 6.1 & 5.0 & 4.8 & 5.1 & 5.6 \\
    & Scaled t5 & 33.9 & 10.7 & 11.5 & 55.5 & 46.3 \\
    & Scaled t3 & 24.0 & 13.5 & 11.0 & 90.6 & 81.1 \\
    \bottomrule
  \end{tabular}
    \caption{Estimated size and power of unconditional discrete Z-tests.\\  We report the percentage of rejections of the null hypothesis at the 5\% confidence level based on \num{65536} replications. The number of days in each backtest sample is $n=750$. The narrow window is [0.985, 0.995] and the wide window is [0.95, 0.995].}
  \label{table:unconditional-discrete}
\end{table}

Table~\ref{table:unconditional-discrete} provides a comparison of the
5-level discrete Z-tests (ZU5 and PE5) with their 3-level counterparts
(ZU3 and PE3) and the
binomial score test (BIN).  For the discrete uniform kernel, the
5-level test is similar in size and power to the 3-level test. For the Pearson kernel,
however, the 5-level test is notably \textit{less} powerful than the 3-level test, and is slightly oversized. 
On this evidence, there appears
to be no advantage in choosing 5-level tests over 3-level tests.

\begin{table}[htbp]
  \centering
  \begin{tabular}{*{3}{l}*{10}{r}}
    \toprule
    window & \( F \) & \( n \) \textbar\ kernel & \multicolumn{1}{c}{BIN} & \multicolumn{1}{c}{ZU3} & \multicolumn{1}{c}{PE3} & \multicolumn{1}{c}{ZU} & \multicolumn{1}{c}{ZA} & \multicolumn{1}{c}{ZE} & \multicolumn{1}{c}{\ZLp} & \multicolumn{1}{c}{\ZLn} & \multicolumn{1}{c}{ZLL} & \multicolumn{1}{c}{PNS} \\
    \midrule
    narrow & Normal & 250 & 4.1 & 4.2 & 5.0 & 3.9 & 3.9 & 3.9 & 4.1 & 3.7 & 5.3 & 5.1 \\
    &  & 500 & 3.9 & 4.6 & 5.4 & 4.6 & 4.6 & 4.5 & 4.6 & 4.6 & 4.7 & 4.7 \\
    &  & 750 & 6.1 & 4.9 & 5.3 & 4.7 & 4.7 & 4.7 & 4.6 & 4.8 & 4.8 & 4.9 \\ \addlinespace[3pt]
    & Scaled t5 & 250 & 17.4 & 19.6 & 18.0 & 18.5 & 18.9 & 18.0 & 22.0 & 14.6 & 20.9 & 22.5 \\
    &  & 500 & 22.1 & 27.1 & 30.9 & 26.5 & 26.9 & 25.7 & 31.5 & 21.6 & 30.2 & 33.6 \\
    &  & 750 & 33.9 & 35.0 & 40.3 & 33.8 & 34.4 & 33.0 & 40.3 & 27.1 & 40.0 & 44.7 \\ \addlinespace[3pt]
    & Scaled t3 & 250 & 13.4 & 15.3 & 17.5 & 14.3 & 14.7 & 13.8 & 19.2 & 9.7 & 20.8 & 22.9 \\
    &  & 500 & 15.9 & 20.2 & 31.8 & 19.6 & 20.1 & 18.7 & 26.4 & 14.0 & 31.0 & 36.7 \\
    &  & 750 & 24.0 & 24.8 & 43.4 & 23.9 & 24.3 & 23.3 & 32.7 & 16.5 & 43.3 & 50.5 \\ \addlinespace[6pt]
    wide & Normal & 250 & 4.1 & 4.4 & 5.2 & 4.8 & 4.8 & 4.8 & 4.7 & 4.8 & 4.8 & 5.1 \\
    &  & 500 & 3.9 & 4.7 & 5.1 & 4.9 & 4.9 & 4.8 & 4.7 & 4.9 & 4.8 & 5.0 \\
    &  & 750 & 6.1 & 5.0 & 5.1 & 4.9 & 4.9 & 4.9 & 4.9 & 4.9 & 5.0 & 5.0 \\ \addlinespace[3pt]
    & Scaled t5 & 250 & 17.4 & 8.1 & 23.0 & 5.9 & 6.3 & 5.7 & 8.9 & 4.9 & 17.2 & 24.4 \\
    &  & 500 & 22.1 & 9.7 & 40.3 & 6.3 & 6.5 & 6.0 & 10.6 & 5.4 & 31.3 & 41.6 \\
    &  & 750 & 33.9 & 10.7 & 55.5 & 6.4 & 6.6 & 6.1 & 11.9 & 5.8 & 45.1 & 57.5 \\ \addlinespace[3pt]
    & Scaled t3 & 250 & 13.4 & 9.1 & 36.1 & 7.7 & 9.1 & 6.8 & 6.3 & 10.9 & 30.2 & 42.7 \\
    &  & 500 & 15.9 & 11.3 & 70.9 & 12.8 & 14.8 & 11.1 & 6.8 & 21.5 & 64.9 & 77.4 \\
    &  & 750 & 24.0 & 13.5 & 90.6 & 17.7 & 20.4 & 15.4 & 7.4 & 31.9 & 85.8 & 93.1 \\
    \bottomrule
  \end{tabular}
  \caption{Effect of backtest sample size on size and power of unconditional Z-tests. We report the percentage of rejections of the null hypothesis at the 5\% confidence level based on \num{65536} replications. The narrow window is [0.985, 0.995] and the wide window is [0.95, 0.995].}
  \label{table:unconditional-sampsize}
\end{table}

Table~\ref{table:unconditional-sampsize} demonstrates that the conclusions drawn from Table~2 in the main paper are robust to the choice of backtest sample size $n$.  As one would expect, power typically increases with $n$.  For both scaled $t$ distributions and both kernel windows, we find that the ordering across kernels in power is little changed. Most importantly, the five properties summarized in Section 4.3 hold regardless of $n$.

\begin{table}[htb!]
  \centering
  \begin{tabular}{*{2}{l}*{7}{r}}
    \toprule
    window & \( F \) \textbar\ test & \multicolumn{1}{c}{BIN} & \multicolumn{1}{c}{LR1} & \multicolumn{1}{c}{ZU3} & \multicolumn{1}{c}{PE3} & \multicolumn{1}{c}{LR3} & \multicolumn{1}{c}{PNS} & \multicolumn{1}{c}{LRB} \\
    \midrule
    narrow & Normal & 6.1 & 4.1 & 4.9 & 5.3 & 8.2 & 4.9 & 5.5 \\
    & Scaled t5 & 33.9 & 24.0 & 35.0 & 40.3 & 34.3 & 44.7 & 37.6 \\
    & Scaled t3 & 24.0 & 16.1 & 24.8 & 43.4 & 46.5 & 50.5 & 49.2 \\ \addlinespace[3pt]
    wide & Normal & 6.1 & 4.1 & 5.0 & 5.1 & 6.1 & 5.0 & 5.1 \\
    & Scaled t5 & 33.9 & 24.0 & 10.7 & 55.5 & 52.2 & 57.5 & 57.7 \\
    & Scaled t3 & 24.0 & 16.1 & 13.5 & 90.6 & 92.7 & 93.1 & 95.0 \\
    \bottomrule
  \end{tabular}
      \caption{Comparison of LR-tests and Z-tests in size and power.\\  We report the percentage of rejections of the null hypothesis at the 5\% confidence level based on \num{65536} replications. The number of days in each backtest sample is $n=750$. The narrow window is [0.985, 0.995] and the wide window is [0.95, 0.995].}
  \label{table:unconditional-comparison}
\end{table}

Table~\ref{table:unconditional-comparison} compares LR-tests with
Z-tests. The binomial LR-test (LR1) is less powerful than the binomial score
test; the former is slightly oversized while the latter is a touch
undersized. The 3-level LR-test (LR3) is compared against a 3-level monospectral test (ZU3) 
and a trispectral Pearson test (PE3). The LR-test is similar to the Pearson test in power (and mostly more
powerful than the ZU3 test), but the LR-test is oversized. 
The generalized Berkowitz test (LRB) performs slightly less well
than the truncated probitnormal score test (PNS) for the narrow window
and very slightly better for the wider window.  Results of a similar exercise with smaller backtest samples ($n=250$) are qualitatively similar (untabulated). We conclude that the Z-tests 
are preferred to their LR-test counterparts, particularly when size is a paramount concern.

The next two experiments relate to conditional Z-tests; the CVT choices are 
defined in Table 3 in the main paper. 
Note that when CVT takes the value \textit{None}, the test is an unconditional
test. Table~\ref{table:dynamic-size} estimates the size of the tests
when samples are uniformly distributed and serially independent. The same messages
emerge as in the excerpt in Section~5.2: when
the CVT is \EMone\ or \EMtwo\ the tests are quite badly oversized,
particularly for the former; choosing \MDfour\ or \MDhalf\ as CVT substantially 
mitigates (but does not eliminate) oversizing.

\begin{table}[htbp]
  \centering
  \begin{tabular}{*{2}{l}*{6}{r}}
    \toprule
    window & CVT \textbar\ kernel & \multicolumn{1}{c}{BIN} & \multicolumn{1}{c}{ZU} & \multicolumn{1}{c}{\ZLp} & \multicolumn{1}{c}{\ZLn} & \multicolumn{1}{c}{ZLL} & \multicolumn{1}{c}{PNS} \\
    \midrule
    narrow & None & 6.2 & 4.8 & 4.6 & 4.9 & 4.9 & 5.0 \\
    & \EMone & 13.3 & 14.4 & 16.0 & 11.5 & 11.8 & 10.4 \\
    & \EMtwo & 8.0 & 9.0 & 10.4 & 8.4 & 8.5 & 8.9 \\
    & \MDfour & 6.8 & 6.7 & 7.2 & 6.4 & 6.6 & 6.4 \\
    & \MDhalf & 6.7 & 6.7 & 7.0 & 6.4 & 6.6 & 6.4 \\ \addlinespace[3pt]
    wide & None & 6.2 & 4.9 & 5.0 & 4.9 & 4.9 & 4.9 \\
    & \EMone & 13.3 & 8.5 & 9.2 & 8.3 & 8.3 & 7.7 \\
    & \EMtwo & 8.0 & 7.3 & 8.1 & 7.0 & 6.9 & 6.7 \\
    & \MDfour & 6.8 & 5.3 & 5.6 & 5.2 & 5.3 & 5.4 \\
    & \MDhalf & 6.7 & 5.5 & 5.7 & 5.3 & 5.5 & 5.6 \\
    \bottomrule
  \end{tabular}
  \caption{Estimated size of tests of conditional coverage.\\   We report the percentage of rejections of the null hypothesis at the 5\% confidence level based on \num{65536} replications. The number of days in each backtest sample is $n=750$. ARMA parameters are AR = 0.95, MA = -0.85. The narrow window is [0.985, 0.995] and the wide window is [0.95, 0.995].} 
  \label{table:dynamic-size}
\end{table}

Table~\ref{table:dynamic-power} is an examination of power for the same
kernels and CVT functions. The aim of the underlying simulation is to
produce pseudo PIT values that are (i) serially dependent with a
dependence structure that is typical when stochastic volatility in the
data is ignored and (ii) possibly non-uniform with the same
distributions used for the unconditional tests.

The data generating process is designed to mimic the behavior of volatile financial return series, such as daily
log-returns on a stock index. Suppose we have empirical return data $X_1,\ldots,X_n$ and we form a
version of the empirical cdf bounded away from 0 and 1 by taking $\widehat{F}_n(x) = (n+1)^{-1} \sum_{t=1}^n
\indicator{X_t \leq x}$ and then construct data $
\widehat{U}_t = \widehat{F}_n(X_t)$.  The transformed returns $
(\widehat{U}_t)$ are close to uniformly distributed and show
negligible serial correlation. However, under the
v-shaped transformation $V(u) = |2u-1|$, we obtain data
$(V(\widehat{U}_t))$ which remain approximately uniformly distributed
but show strong serial correlation. Let $T$
denote the transformation $T(u) = \Phi^{-1}(V(u))$. If we fit
Gaussian ARMA models to the transformed data $(T(\widehat{U}_t))$ we
find that an ARMA(1,1) process often fits well and typical values for the AR and MA parameters are
around 0.95 and -0.85.

We want to generate losses $(L_t)$ with marginal
distribution $\losscdf$ in Table~\ref{table:dynamic-power} such that,
if $U_t = \losscdf(L_t)$,
the process
$(T(U_t))$ is \textit{exactly} a Gaussian ARMA(1,1) process with AR parameter 0.95, MA parameter -0.85, mean zero and
variance one. Let  $(Z_t)$ be such a Gaussian ARMA process and let $(\cointoss_t)$
be a series of iid Bernoulli variables with mean
0.5. We apply the following series of transformations to construct $(L_t)$:
\begin{equation}\label{eq:55}
  \tilde{U}_t = \Phi(Z_t),\quad U_t =
\tfrac{1}{2}(1+\tilde{U}_t)^{\cointoss_t}(1-\tilde{U}_t)^{(1-\cointoss_t)},\quad L_t = \losscdf^{-1}(U_t).
\end{equation}
The first transformation induces uniformity; the second can be thought of as a method of
stochastically inverting $V(u) = |2u-1|$
while preserving uniformity; the third transformation gives losses with df $\losscdf$. As for the unconditional tests, pseudo PIT values are
obtained by the transformation $P_t = \Phi(L_t)$.

% Let $T(u) = \Phi^{-1}(V(\losscdf(u)))$ define the transformation obtained by
% concatenating $\losscdf$, $V$ and $\Phi^{-1}$.
% The DGP has the property that the transformed losses $(T(L_t))$ follow a Gaussian
% ARMA(1,1) process with the given parameters. This mimics the behavior
% of volatile financial return series, such as daily
% log-returns on a stock index. Now let $\widehat{T}(u) = \Phi^{-1}(V(\widehat{F}_n(u)))$. Then the serial dependence in the transformed data
% $(\widehat{T}(X_t))$ can often be modelled by an
% ARMA(1,1) model with typical values for the AR and MA parameters 
% around 0.95 and -0.85.

\begin{table}[htbp]
  \centering
  \begin{tabular}{*{3}{l}*{6}{r}}
    \toprule
    window & \( $\losscdf$ \) & CVT \textbar\ kernel & \multicolumn{1}{c}{BIN} & \multicolumn{1}{c}{ZU} & \multicolumn{1}{c}{\ZLp} & \multicolumn{1}{c}{\ZLn} & \multicolumn{1}{c}{ZLL} & \multicolumn{1}{c}{PNS} \\
    \midrule
   narrow & Normal & None & 12.1 & 10.8 & 10.0 & 11.2 & 9.1 & 9.5 \\
    &  & \EMone & 30.0 & 31.5 & 32.0 & 29.4 & 29.1 & 28.1 \\
    &  & \EMtwo & 28.1 & 30.9 & 31.8 & 30.2 & 29.5 & 30.5 \\
    &  & \MDfour & 30.3 & 32.6 & 30.7 & 33.6 & 32.3 & 33.4 \\
    &  & \MDhalf & 19.3 & 21.7 & 19.9 & 22.5 & 22.0 & 23.2 \\ \addlinespace[3pt]
    & Scaled t5 & None & 36.0 & 36.2 & 40.9 & 31.4 & 41.2 & 45.1 \\
    &  & \EMone & 47.4 & 54.9 & 59.7 & 46.1 & 53.4 & 53.4 \\
    &  & \EMtwo & 48.4 & 52.7 & 56.8 & 49.5 & 54.7 & 56.4 \\
    &  & \MDfour & 56.4 & 60.7 & 63.1 & 57.2 & 61.1 & 62.0 \\
    &  & \MDhalf & 49.6 & 54.5 & 57.3 & 50.5 & 55.6 & 56.8 \\ \addlinespace[3pt]
    & Scaled t3 & None & 28.1 & 28.3 & 35.0 & 22.2 & 44.1 & 50.7 \\
    &  & \EMone & 42.2 & 50.4 & 56.1 & 39.7 & 53.3 & 54.3 \\
    &  & \EMtwo & 42.5 & 47.3 & 53.5 & 42.6 & 53.3 & 55.6 \\
    &  & \MDfour & 50.1 & 54.8 & 58.9 & 49.7 & 58.8 & 60.4 \\
    &  & \MDhalf & 44.1 & 49.5 & 54.1 & 43.7 & 54.2 & 56.2 \\ \addlinespace[6pt]
    wide & Normal & None & 12.1 & 17.8 & 15.9 & 18.3 & 14.6 & 15.4 \\
    &  & \EMone & 30.0 & 31.1 & 30.1 & 31.6 & 30.4 & 30.3 \\
    &  & \EMtwo & 28.1 & 36.0 & 34.5 & 36.2 & 34.6 & 34.9 \\
    &  & \MDfour & 30.3 & 52.1 & 46.4 & 54.1 & 51.2 & 53.8 \\
    &  & \MDhalf & 19.3 & 44.9 & 36.9 & 48.0 & 44.7 & 48.6 \\ \addlinespace[3pt]
    & Scaled t5 & None & 36.0 & 19.5 & 22.3 & 19.2 & 52.9 & 63.9 \\
    &  & \EMone & 47.4 & 35.4 & 38.9 & 31.5 & 49.8 & 57.9 \\
    &  & \EMtwo & 48.4 & 41.0 & 45.3 & 37.7 & 55.1 & 62.6 \\
    &  & \MDfour & 56.4 & 55.2 & 56.7 & 52.8 & 67.2 & 74.1 \\
    &  & \MDhalf & 49.6 & 51.1 & 51.4 & 49.2 & 65.7 & 73.5 \\ \addlinespace[3pt]
    & Scaled t3 & None & 28.1 & 28.4 & 18.5 & 39.3 & 87.6 & 93.9 \\
    &  & \EMone & 42.2 & 32.8 & 32.6 & 34.2 & 71.1 & 82.1 \\
    &  & \EMtwo & 42.5 & 38.8 & 38.7 & 40.1 & 75.3 & 85.2 \\
    &  & \MDfour & 50.1 & 50.7 & 48.4 & 52.7 & 82.5 & 90.2 \\
    &  & \MDhalf & 44.1 & 48.1 & 43.4 & 51.2 & 83.6 & 91.1 \\
    \bottomrule
  \end{tabular}
  \caption{Estimated power of tests of conditional coverage when DGP is based on an ARMA(1,1) model.\\   We report the percentage of rejections of the null hypothesis at the 5\% confidence level based on \num{65536} replications. The number of days in each backtest sample is $n=750$. ARMA parameters are AR = 0.95, MA = -0.85. The narrow window is [0.985, 0.995] and the wide window is [0.95, 0.995].} 
  \label{table:dynamic-power}
\end{table}

In Table~\ref{table:dynamic-power} we observe that there is generally
a very large increase in power when we move from the unconditional
tests (CVT = \textit{None}) to the conditional tests. This is evident even when
the distribution of the simulated data is uniform ($\losscdf$ = Normal). The
most powerful tests are bispectral tests applied to the wider window
using the CVT functions \MDfour\ and \MDhalf.

\end{onehalfspacing}

\bibliographystyle{jf}  
\bibliography{spectraltest}